\def \VersionLong {}
\ifdefined \VersionLong
	\newcommand{\LongVersion}[1]{\ifdefined\VersionWithComments{\color{red!40!black}#1}\else#1\fi}
	\newcommand{\ShortVersion}[1]{\ifdefined\VersionWithComments{\color{black!40}#1}\fi}
\else
	\newcommand{\LongVersion}[1]{\ifdefined\VersionWithComments{\color{black!60}#1}\fi}
	\newcommand{\ShortVersion}[1]{\ifdefined\VersionWithComments{\color{red!40!black}#1}\else#1\fi}
\fi

\newcommand{\LongVersionTable}[1]{}
\newcommand{\ShortVersionTable}[1]{#1}
\newcommand{\nbLoC}[1]{\cellcolor{gray!20}}

\documentclass[a4paper,10pt]{llncs}
\ifdefined \VersionLong
\else
\makeatletter
\AtBeginDocument{%
  \@ifpackageloaded{hyperref}
  {\def\@doi#1{\href{https://doi.org/#1}
      {\ttfamily https://doi.org/#1}\egroup}}
  {\def\@doi#1{\ttfamily https://doi.org/#1\egroup}}
  \def\doi{\bgroup\catcode`\_=12\relax\@doi}}
\makeatother
\fi

\usepackage[utf8]{inputenc}
\usepackage[english]{babel}

\usepackage[ruled,vlined,linesnumbered]{algorithm2e}
	\SetKwInOut{Input}{input}
	\SetKwInOut{Output}{output}

\usepackage{subcaption}

\usepackage{paralist} %

\newenvironment{ienumeration}
	{\ifdefined\VersionLong\begin{enumerate}\else\begin{inparaenum}[\itshape i\upshape)]\fi}
	{\ifdefined\VersionLong\end{enumerate}\else\end{inparaenum}\fi}

\usepackage{amsmath} %
\usepackage{amssymb} %

\ifdefined\VersionWithComments
	\usepackage{draftwatermark}
	\SetWatermarkText{draft}
	\SetWatermarkScale{15}
	\SetWatermarkColor[gray]{0.9}
\fi

\usepackage[firstpage]{draftwatermark}
\SetWatermarkText{\hspace*{135mm}\raisebox{152mm}{\includegraphics[scale=0.12]{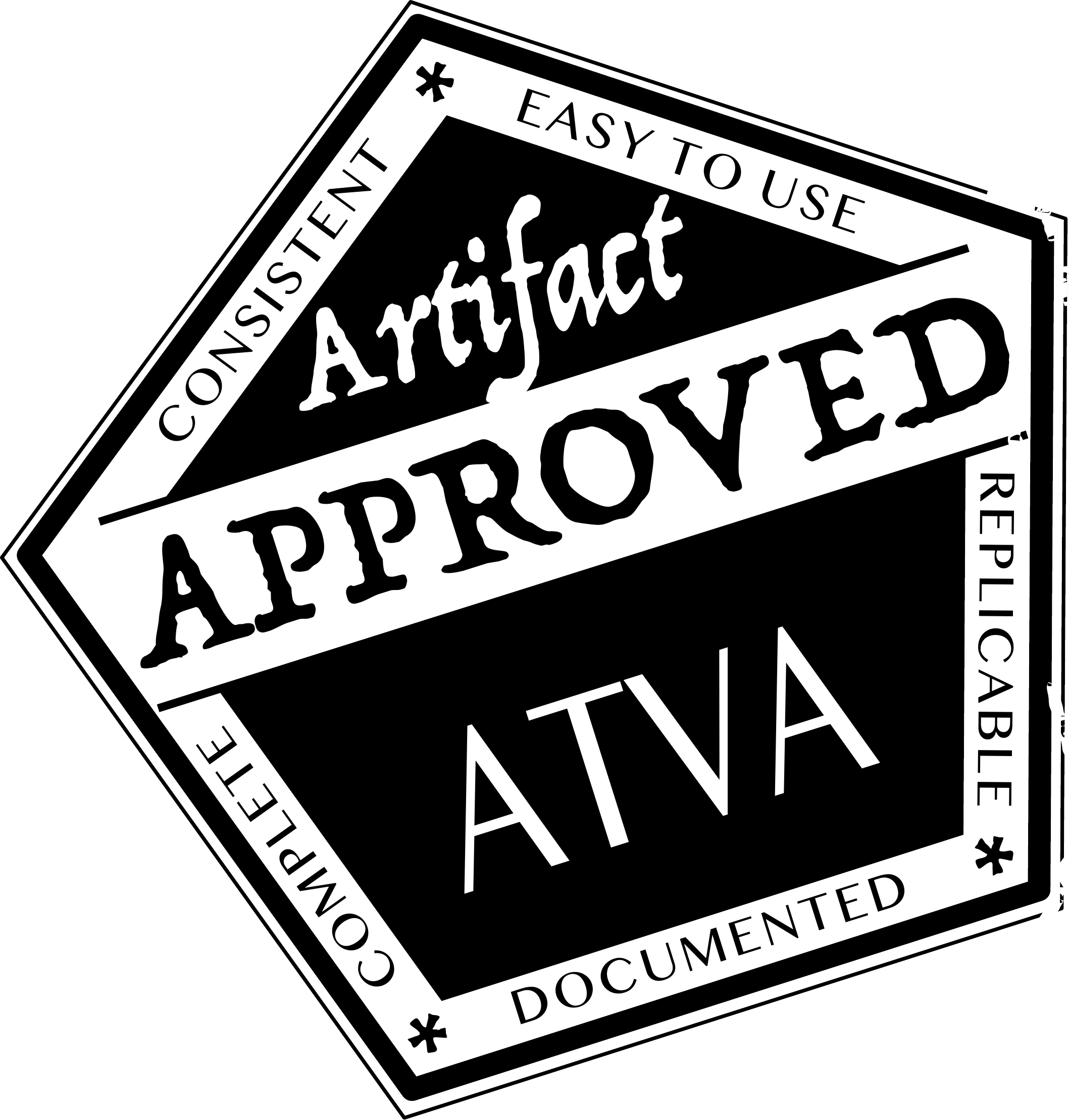}}}
\SetWatermarkAngle{0}

\usepackage[svgnames,table]{xcolor}
\definecolor{darkblue}{rgb}{0, 0, 0.7}

\usepackage[
		pdfauthor={Étienne André and Sun Jun},%
		pdftitle={Parametric Timed Model Checking for Guaranteeing Timed Opacity},
		breaklinks  = true,
		colorlinks  = true,
		citecolor   = blue!50!blue,
		linkcolor   = darkblue,
		urlcolor    = blue!50!green,
	]{hyperref}

\usepackage[capitalise,english,nameinlink]{cleveref} %
\crefname{line}{\text{line}}{\text{lines}} %

\usepackage{listings}

\definecolor{mygreen}{rgb}{0,0.6,0}
\definecolor{mygray}{rgb}{0.5,0.5,0.5}
\definecolor{mymauve}{rgb}{0.58,0,0.82}

\newcommand{\defProblem}[3]
{%
	\noindent\fcolorbox{black}{blue!15}{
	\begin{minipage}{.95\columnwidth}
		\textbf{#1 Problem:}\\
		\textsc{Input}: #2\\
		\textsc{Problem}: #3
	\end{minipage}
}
	
	\smallskip
	
}

\usepackage{tikz}
\usetikzlibrary{arrows,automata,positioning,shapes} %
\tikzstyle{every node}=[initial text=]
\tikzstyle{location}=[rectangle, rounded corners, minimum size=12pt, draw=black, fill=blue!10, inner sep=2pt]
\tikzstyle{invariant}=[draw=black, dotted, inner sep=1pt, node distance=0] %
\tikzstyle{final}=[double, fill=blue!50]

\tikzstyle{urgent}=[fill=yellow, thick, dotted] %
\tikzstyle{private}=[fill=red,thick]

\definecolor{coloract}{rgb}{0.50, 0.70, 0.30}
\definecolor{colorclock}{rgb}{0.4, 0.4, 1}
\definecolor{colordisc}{rgb}{1, 0, 1}
\definecolor{colorloc}{rgb}{0.4, 0.4, 0.65}
\definecolor{colorparam}{rgb}{1, 0.6, 0.0}

\definecolor{loccolor1}{rgb}{1, 0.3, 0.3}
\definecolor{loccolor2}{rgb}{0.3, 1, 0.3}
\definecolor{loccolor3}{rgb}{0.3, 0.3, 1}
\definecolor{loccolor4}{rgb}{1, 0.3, 1}
\definecolor{loccolor5}{rgb}{1, 1, 0.3}
\definecolor{loccolor6}{rgb}{0.3, 1, 1}
\definecolor{loccolor7}{rgb}{0.9, 0.6, 0.2}
\definecolor{loccolor8}{rgb}{0.7, 0.4, 1}
\definecolor{loccolor9}{rgb}{0.5, 1, 0.75}
\definecolor{loccolor10}{rgb}{0.8, 0.7, 0.6}
\definecolor{loccolor11}{rgb}{0.6, 0.7, 0.8}
\definecolor{loccolor12}{rgb}{0.2, 0.5, 0.9}
\definecolor{loccolor13}{rgb}{0.5, 0.9, 0.2}
\definecolor{loccolor14}{rgb}{0.9, 0.2, 0.5}
\definecolor{loccolor15}{rgb}{0.7, 0.7, 0.7}
\definecolor{loccolor16}{rgb}{0.8, 0.8, 0.5}

\newcommand{\styleact}[1]{\ensuremath{\textcolor{coloract}{\mathrm{#1}}}}
\newcommand{\styleclock}[1]{\ensuremath{\textcolor{colorclock}{\mathrm{#1}}}}
\newcommand{\styledisc}[1]{\ensuremath{\textcolor{colordisc}{\mathrm{#1}}}}
\newcommand{\styleloc}[1]{\ensuremath{\mathrm{#1}}}
\newcommand{\styleparam}[1]{\ensuremath{\textcolor{colorparam}{\mathrm{#1}}}}

\newcommand{\stylecode}[1]{\textcolor{colorloc}{\texttt{#1}}}

\newcommand{\stylebench}[1]{\textcolor{colorloc}{\texttt{#1}}}

\newcommand{\cellHeader}[0]{\cellcolor{blue!20}\bfseries}
\newcommand{\rowHeader}{\rowcolor{blue!20}\bfseries}
\newcommand{\cellYes}{\cellcolor{red!20}\textbf{$\surd$}}
\newcommand{\cellNo}{\cellcolor{green!20}\textbf{$\times$}}
\newcommand{\cellFixable}{\cellcolor{orange!20}\textbf{$(\surd)$}}

\newcommand{\cellKall}{\cellcolor{blue!20}\textbf{$\KTrue$}}
\newcommand{\cellKnone}{\cellcolor{red!20}\textbf{$\KFalse$}}
\newcommand{\cellKsome}{\cellcolor{green!20}\textbf{$K$}}

\newcommand{\init}{_0}

\newcommand{\A}{\ensuremath{\mathcal{A}}}
\newcommand{\Azeroinf}{\ensuremath{\A_{0,\infty}}}
\newcommand{\Actions}{\Sigma}
\newcommand{\action}{\ensuremath{a}}
\newcommand{\actionEnd}{\ensuremath{\styleact{finish}}}
\newcommand{\ActionsIndices}{\zeta}
\newcommand{\assign}{\leftarrow}
\newcommand{\bflag}{\ensuremath{b}} %

\newcommand{\BTrue}{\text{true}}
\newcommand{\BFalse}{\text{false}}
\newcommand{\C}{C}

\newcommand{\Clock}{\mathbb{X}} %
\newcommand{\ClockCard}{H} %
\newcommand{\clock}{x} %
\newcommand{\clockabs}{\ensuremath{x_\mathit{abs}}} %
\newcommand{\clockval}{\mu} %
\newcommand{\ClocksZero}{\vec{0}}
\newcommand{\compOp}{\bowtie}

\newcommand{\compOpLeq}{\triangleleft}
\newcommand{\CTrue}{\mathbf{true}}

\newcommand{\duration}{\ensuremath{\mathit{dur}}}
\newcommand{\edge}{e}
\newcommand{\Edges}{E}
\newcommand{\longuefleche}[1]{\stackrel{#1}{\longrightarrow}}
\newcommand{\longueflecheRel}[1]{\stackrel{#1}{\mapsto}}

\newcommand{\flecheRel}{{\rightarrow}}

\newcommand{\grandn}{{\mathbb N}}
\newcommand{\grandq}{{\mathbb Q}}
\newcommand{\grandqplus}{\grandq_{+}} %
\newcommand{\grandr}{\ensuremath{\mathbb R}}
\newcommand{\grandrplus}{\ensuremath{\grandr_{+}}} %
\newcommand{\grandz}{{\mathbb Z}}
\newcommand{\guard}{g}

\newcommand{\invariant}{I}
\newcommand{\K}{K}
\newcommand{\KTrue}{\top}
\newcommand{\KFalse}{\bot}

\newcommand{\loc}{\ensuremath{\ell}} %
\newcommand{\locinit}{\loc\init}
\newcommand{\Loc}{L} %
\newcommand{\locfinal}{\ensuremath{\loc_f}}
\newcommand{\locpriv}{\ensuremath{\loc_{\mathit{priv}}}}
\newcommand{\locpub}{\ensuremath{\loc_{\mathit{pub}}}}

\newcommand{\lterm}{\mathit{lt}}
\newcommand{\Param}{\mathbb{P}} %
\newcommand{\param}{p} %
\newcommand{\paramabs}{\ensuremath{\param_\mathit{abs}}} %
\newcommand{\ParamCard}{M} %
\newcommand{\pval}{v} %
\newcommand{\PZG}{\ensuremath{\mathcal{PZG}}} %
\newcommand{\R}{{\mathbb{R}}}
\newcommand{\Rgeqzero}{\R_{\geq 0}}

\newcommand{\sinit}{s\init} %
\newcommand{\somelocs}{T} %
\newcommand{\state}{\ensuremath{s}} %
\newcommand{\States}{S} %
\newcommand{\Succ}{\mathsf{Succ}}
\newcommand{\timelapse}[1]{#1^\nearrow}
\newcommand{\Times}{\ensuremath{D}}
\newcommand{\varrun}{\rho} %

\newcommand{\PrivDurReach}[3]{\ensuremath{\mathit{DReach}^{#1}_{#2}(#3)}}
\newcommand{\PubDurReach}[3]{\ensuremath{\mathit{DReach}^{#1}_{\neg #2}(#3)}}

\newcommand{\styleSymbStatesSet}[1]{\ensuremath{\mathbf{#1}}}

\newcommand{\symbstate}{\ensuremath{\styleSymbStatesSet{s}}} %
\newcommand{\SymbState}{\ensuremath{\styleSymbStatesSet{S}}} %
\newcommand{\symbstateinit}{\symbstate\init} %
\newcommand{\symbtrans}{{\Rightarrow}} %

\newcommand{\resets}{R}
\newcommand{\projectP}[1]{\ensuremath{#1{\downarrow_{\Param}}}}
\newcommand{\reset}[2]{\ensuremath{[#1]_{#2}}}
\newcommand{\valuate}[2]{\ensuremath{#2(#1)}}
\newcommand{\wv}[2]{#1|#2} %

\newcommand{\stylealgo}[1]{\ensuremath{\textsf{#1}}}
\newcommand{\Copy}{\stylealgo{Copy}}
\newcommand{\EFsynth}{\stylealgo{EFsynth}}
\newcommand{\Enrich}{\stylealgo{Enrich}}
\newcommand{\SynthOp}{\stylealgo{SynthOp}}

\ifdefined\VersionWithComments
	\usepackage[colorinlistoftodos,textsize=footnotesize]{todonotes}
\else
	\usepackage[disable]{todonotes}
\fi
\newcommand{\gennote}[3]{\todo[linecolor=#2,backgroundcolor=#2!25,bordercolor=#2]{#3: #1}}
\newcommand{\ea}[1]{\gennote{#1}{blue}{ÉA}}
\newcommand{\sj}[1]{{\gennote{#1}{purple}{SJ}}}
\newcommand{\instructions}[1]{{\gennote{\bfseries #1}{red}{Instructions}}}

\ifdefined \VersionWithComments
	\newcommand{\todoinline}[1]{\mbox{}{\color{red}{\textbf{TODO}\ifx#1\\\else:\ \fi #1}}} %
\else
	\newcommand{\todoinline}[1]{}
\fi

\usepackage{verbatim} %

\LongVersion{
	\usepackage[backend=biber,backref=true,style=alphabetic,url=false,doi=true,defernumbers=true,sorting=anyt,maxnames=99]{biblatex} %
	\addbibresource{PTA.bib}

	\renewbibmacro*{doi+eprint+url}{%
		\iftoggle{bbx:doi}
			{\color{black!40}\footnotesize\printfield{doi}}
			{}%
		\newunit\newblock
		\iftoggle{bbx:eprint}
			{\usebibmacro{eprint}}
			{}%
		\newunit\newblock
		\iftoggle{bbx:url}
			{\usebibmacro{url+urldate}}
			{}%
	}
	
}

\newcommand{\imitator}{\textsf{IMITATOR}}
\newcommand{\uppaal}{\textsc{Uppaal}}

\ifdefined \VersionWithComments
 	\definecolor{colorok}{RGB}{80,80,150}
\else
	\definecolor{colorok}{RGB}{0,0,0}
\fi

\newcommand{\eg}{\textcolor{colorok}{e.\,g.,}\xspace}
\newcommand{\ie}{\textcolor{colorok}{i.\,e.,}\xspace}
\newcommand{\st}{\textcolor{colorok}{s.t.}\xspace}

\newcommand{\wrt}{\textcolor{colorok}{w.r.t.}\xspace}

\renewcommand{\orcidID}[1]{\href{https://orcid.org/#1}{\includegraphics[width=1em]{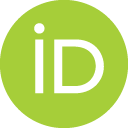}}}

\title{Parametric Timed Model Checking for Guaranteeing Timed Opacity\thanks{%
	\LongVersion{This is the author (and extended) version of the manuscript of the same name published in the proceedings of the 17th International Symposium on Automated Technology for Verification and Analysis (\href{http://atva2019.iis.sinica.edu.tw/}{ATVA 2019}).
	This version contains some additional explanations and all proofs.
	The published version is available at
		\href{https://www.doi.org/10.1007/978-3-030-31784-3_7}{\nolinkurl{10.1007/978-3-030-31784-3_7}}.
	}%
	This work is partially supported by
		the ANR national research program PACS (ANR-14-CE28-0002),
		the ANR-NRF research program ProMiS,
		and
	by ERATO HASUO Metamathematics for Systems Design Project (No.\ JPMJER1603), JST.
}}
\author{\'Etienne Andr\'e\inst{1,2,3}\orcidID{0000-0001-8473-9555}
\and
Jun Sun\inst{4}\orcidID{0000-0002-3545-1392} %
}
\institute{Université Paris 13, LIPN, CNRS, UMR 7030\\F-93430, Villetaneuse, France
\and
JFLI, CNRS, Tokyo, Japan
\and
National Institute of Informatics, Tokyo, Japan
\and
School of Information Systems, Singapore Management University
}

\begin{document}

\pagestyle{plain}

\maketitle

\thispagestyle{plain}

\setcounter{footnote}{0}

\ifdefined \VersionWithComments
	\textcolor{red}{\textbf{This is the version with comments. To disable comments, comment out line~3 in the \LaTeX{} source.}}
\fi

\begin{abstract}
	Information leakage can have dramatic consequences on systems security.
	Among harmful information leaks, the timing information leakage is the ability for an attacker to deduce internal information depending on the system execution time.
	We address the following problem:
		given a timed system, %
			synthesize the execution times %
				for which one cannot deduce whether the system performed some secret behavior.
	We solve this problem in the setting of timed automata (TAs).
	We first provide a general solution, and then extend the problem to parametric TAs, by synthesizing internal timings making the TA secure.
	We study decidability, devise algorithms, and show that our method can also apply %
	to program analysis.
	
	\LongVersion{}
\end{abstract}

\ea{hello}
\sj{hello}

\todo{(Notes from August 2018)
Sun Jun says that the property could be:
	for all output (i.e., a "valuation of variable" (that could use actions or locations)), it must not be possible to know whether we passed by secret locations or not

This paper would be more "program-oriented" than the other

Then, we need:
	examples!
	a small text explaining how we convert a program to a TA. Basically all instructions are 0- or constant-time, while sleep/wait have different times

I'd like POST; Sun Jun says he'd like a stronger one. TACAS ?
FM would be great (but April 2019)
So we target TACAS.
}

\instructions{ATVA 2019:
Regular research papers (16 pages, including references)
}

\section{Introduction}\label{section:introduction}

Timed systems combine concurrency and possibly hard real-time constraints.
Information leakage can have dramatic consequences on the security of such systems.
Among harmful information leaks, the \emph{timing information leakage} is the ability for an attacker to deduce internal information depending on timing information.
In this work, we focus on the execution time, \ie{} when a system works as an almost black-box, with the ability of an attacker to mainly observe its execution time.

We address the following problem:
	given a timed system, a private state denoting the execution of some secret behavior and a final state denoting the completion of the execution,
		synthesize the execution times to the final state for which one cannot deduce whether the system has passed through the private state.
We solve this problem in the setting of timed automata (TAs), which is a popular extension of finite-state automata with clocks~\cite{AD94}.
We first prove that this problem is solvable\ea{i gave up exact complexity}, and we provide an algorithm, that we implement and apply to a set of benchmarks containing notably a set of Java programs known for their (absence of) timing information leakage.

Then we consider a higher-level problem by allowing (internal) timing parameters in the system, that can model uncertainty or unknown constants at early design stage.
The setting becomes parametric timed automata~\cite{AHV93}, and the problem asks:
	given a timed system with timing parameters, a private state and a final state,
		synthesize the timing parameters and the execution times for which one cannot deduce whether the system has passed through the private state.
Although we show that the problem is in general undecidable, we provide a decidable subclass; then we devise a general procedure not guaranteed to terminate, but that behaves well on examples from the literature.
\LongVersion{
\paragraph*{Outline}
After reviewing related works in \cref{section:related},
\cref{section:preliminaries} recalls necessary concepts and \cref{section:problem} introduces the problem.
\cref{section:TA} addresses timed-opacity for timed automata.
We then address the parametric version of timed-opacity, with theory studied in \cref{section:theory},
algorithmic in~\cref{section:synthesis}
and experiments in \cref{section:experiments}.
\cref{section:conclusion} concludes the paper.
}
\section{Related works}\label{section:related}

This work is closely related to the line of work on defining and analyzing information flow in timed automata.
It is well-known (see \eg{} \LongVersion{\cite{Kocher96,FS00,BB07,KPJJ13,BCLR15}}\ShortVersion{\cite{Kocher96,BCLR15}}) that time is a potential attack vector against secure systems.
That is, it is possible that a non-interferent (secure) system can become interferent (insecure) when timing constraints are added~\cite{GMR07}.
In~\LongVersion{\cite{BDST02,BT03}}\ShortVersion{\cite{BDST02}}, a first notion of \emph{timed} non-interference is proposed.
In~\cite{GMR07}, Gardey \emph{et al.}\ define timed strong non-deterministic non-interference (SNNI) based on timed language equivalence between the automaton with hidden low-level actions and the automaton with removed low-level actions. Furthermore, they show that the problem of determining whether a timed automaton satisfies SNNI is undecidable. In contrast, timed cosimulation-based SNNI, timed bisimulation-based SNNI and timed state SNNI are decidable.
In~\cite{Cassez09}, the problem of checking opacity for timed automata is considered: even for the restricted class of event-recording automata\LongVersion{~\cite{AFH99}}, it is undecidable whether a system is opaque, \ie{} whether an attacker can deduce whether some set of actions was performed, by only observing a given set of observable actions (with their timing).
In~\cite{VNN18}, Vasilikos \emph{et al.}\ define the security of timed automata in term of information flow using a bisimulation relation and develop an algorithm for deriving a sound constraint for satisfying the information flow property locally based on relevant transitions.
In~\cite{BCLR15}, Benattar \emph{et al.}\ study the control synthesis problem of timed automata for SNNI. That is, given a timed automaton, they propose a method to automatically generate a (largest) sub-systems such that it is non-interferent if possible.
Different from the above-mentioned work, our work considers parametric timed automata, \ie{} timed systems with unknown design parameters, and focuses on synthesizing parameter valuations which guarantee information flow property.
As far as we know, this is the first work on parametric model checking for timed automata for information flow property.
Compared to~\cite{BCLR15}, our approach is more realistic as it does not require change of program structure.
Rather, our result provides guidelines on how to choose the timing parameters (\eg{} how long to wait after certain program statements) for avoiding information leakage.

In~\cite{NNV17}, the authors propose a type system dealing with non-determinism and (continuous) real-time, the adequacy of which is ensured using non-interference.
We share the common formalism of TA; however, we mainly focus on leakage as execution time, and we \emph{synthesize} internal parts of the system (clock guards), in contrast to~\cite{NNV17} where the system is fixed.

This work is related to work on mitigating information leakage through time side channel\LongVersion{~\cite{DBLP:conf/popl/Agat00,DBLP:conf/icisc/MolnarPSW05,DBLP:conf/sp/CoppensVBS09,DBLP:journals/siglog/WangS17,DBLP:conf/issta/WuGS018}}.
\LongVersion{%
	In~\cite{DBLP:conf/popl/Agat00}, Agat \emph{et al.}\ proposed to eliminate time side channel through type-driven cross-copying.
	In~\cite{DBLP:conf/icisc/MolnarPSW05}, Molnar~\emph{et al.}\ proposed, along the program counter model, a method for mitigating side channel through merging branches. A similar idea was proposed in~\cite{DBLP:journals/entcs/BartheRW06}.
	Coppens \emph{et al.}~\cite{DBLP:conf/sp/CoppensVBS09} developed a compiler backend for removing such leaks on x86 processors.
}%
\ShortVersion{For example, in}\LongVersion{In}~\cite{DBLP:journals/siglog/WangS17}, Wang \emph{et al.}\ proposed to automatically generate masking code for eliminating side channel through program synthesis.
In~\cite{DBLP:conf/issta/WuGS018}, Wu \emph{et al.}\ proposed to eliminate time side channel through program repair.
Different from the above-mentioned works, we reduce the problem of mitigating time side channel as a parametric model checking problem and solve it using parametric reachability analysis techniques. 

This work is related to work on identifying information leakage through timing analysis\LongVersion{~\cite{DBLP:conf/kbse/SungPW18,DBLP:conf/rtss/ChattopadhyayR11,ALKH16,DBLP:conf/cav/ZhangGSW18,DBLP:conf/atal/DennisSF16,DBLP:conf/uss/DoychevFKMR13,DBLP:journals/corr/abs-1807-03280}}.
In~\cite{DBLP:conf/rtss/ChattopadhyayR11}, Chattopadhyay \emph{et al.}\ applied model checking to perform cache timing analysis. In~\cite{DBLP:conf/rtas/ChuJM16}, Chu \emph{et al.}\ performed similar analysis through symbolic execution. In~\cite{ALKH16}, Abbasi \emph{et al.}\ apply the NuSMV model checker to verify integrated circuits against information leakage through side channels. In~\cite{DBLP:conf/uss/DoychevFKMR13}, a tool is developed to identify time side channel through static analysis. In~\cite{DBLP:conf/cav/ZhangGSW18}, Sung \emph{et al.}\ developed a framework based on LLVM for cache timing analysis.

\section{Preliminaries}\label{section:preliminaries}
\LongVersion{In this work, we assume a system is modeled in the form of a parametric timed automaton.}
\ea{Perhaps useless?: \LongVersion{%
	In \cref{ss:Java2PTA}, we discuss how we can model programs with unknown design parameters (\eg{} a Java program with a statement \stylecode{Thread.sleep(n)} where \stylecode{n} is unknown) as parametric timed automata.
}}

\LongVersion{
\subsection{Clocks, parameters and guards}
}

We assume a set~$\Clock = \{ \clock_1, \dots, \clock_\ClockCard \} $ of \emph{clocks}, \ie{} real-valued variables that evolve at the same rate.
A clock valuation is\LongVersion{ a function}
$\clockval : \Clock \rightarrow \Rgeqzero$.
We write $\ClocksZero$ for the clock valuation assigning $0$ to all clocks.
Given $d \in \Rgeqzero$, $\clockval + d$ \ShortVersion{is}\LongVersion{denotes the valuation} \st{} $(\clockval + d)(\clock) = \clockval(\clock) + d$, for all $\clock \in \Clock$.
Given $\resets \subseteq \Clock$, we define the \emph{reset} of a valuation~$\clockval$, denoted by $\reset{\clockval}{\resets}$, as follows: $\reset{\clockval}{\resets}(\clock) = 0$ if $\clock \in \resets$, and $\reset{\clockval}{\resets}(\clock)=\clockval(\clock)$ otherwise.

We assume a set~$\Param = \{ \param_1, \dots, \param_\ParamCard \} $ of \emph{parameters}\LongVersion{, \ie{} unknown constants}.
A parameter {\em valuation} $\pval$ is\LongVersion{ a function}
$\pval : \Param \rightarrow \grandqplus$.
We assume ${\compOp} \in \{<, \leq, =, \geq, >\}$.
A guard~$\guard$ is a constraint over $\Clock \cup \Param$ defined by a conjunction of inequalities of the form
$\clock \compOp \sum_{1 \leq i \leq \ParamCard} \alpha_i \param_i + d$, with
	$\param_i \in \Param$,
	and
	$\alpha_i, d \in \grandz$.
Given~$\guard$, we write~$\clockval\models\pval(\guard)$ if %
the expression obtained by replacing each~$\clock$ with~$\clockval(\clock)$ and each~$\param$ with~$\pval(\param)$ in~$\guard$ evaluates to true.

\LongVersion{
\subsection{Parametric timed automata}

Parametric timed automata (PTA) extend timed automata with parameters within guards and invariants in place of integer constants~\cite{AHV93}.
}

\begin{definition}[PTA]\label{def:uPTA}
	A PTA $\A$ is a tuple \mbox{$\A = (\Actions, \Loc, \locinit, \Clock, \Param, \invariant, \Edges)$}, where:
	\begin{ienumeration}
		\item $\Actions$ is a finite set of actions,
		\item $\Loc$ is a finite set of locations,
		\item $\locinit \in \Loc$ is the initial location,
		\item $\Clock$ is a finite set of clocks,
		\item $\Param$ is a finite set of parameters,
		\item $\invariant$ is the invariant, assigning to every $\loc\in \Loc$ a guard $\invariant(\loc)$,
		\item $\Edges$ is a finite set of edges  $\edge = (\loc,\guard,\action,\resets,\loc')$
		where~$\loc,\loc'\in \Loc$ are the source and target locations, $\action \in \Actions$, $\resets\subseteq \Clock$ is a set of clocks to be reset, and $\guard$ is a guard.
	\end{ienumeration}
\end{definition}
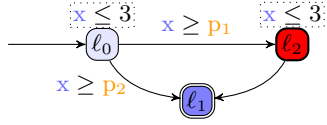
\begin{figure}[tb]
 
	\centering
	 \footnotesize

	\begin{tikzpicture}[scale=1, xscale=2.5, yscale=1.5, auto, ->, >=stealth']
 
		\node[location, initial] at (0, 0) (s0) {$\loc_0$};
 
		\node[location, private] at (1, 0) (s2) {$\loc_2$};
 
		\node[location, final] at (.5, -.5) (s1) {$\loc_1$};
 
		\node[invariant, above=of s0] {$\styleclock{\clock} \leq 3$};
		\node[invariant, above=of s2] {$\styleclock{\clock} \leq 3$};
		
		\path (s0) edge node[align=center]{$\styleclock{\clock} \geq \styleparam{\param_1}$} (s2);
		\path (s0) edge[bend right] node[left, align=center]{$\styleclock{\clock} \geq \styleparam{\param_2}$} (s1);
		\path (s2) edge[bend left] node[align=center]{} (s1);

	\end{tikzpicture}
	\caption{A PTA example}
	\label{figure:example-PTA}

\end{figure}
\begin{example}
	Consider the PTA in \cref{figure:example-PTA} (inspired by \cite[Fig.~1b]{GMR07}), containing one clock~$\clock$ and two parameters~$\param_1$ and~$\param_2$.
	$\loc_0$ is the initial location, while~$\loc_1$ is the (only) accepting location.
\end{example}

Given\LongVersion{ a parameter valuation}~$\pval$, we denote by $\valuate{\A}{\pval}$ the non-parametric structure where all occurrences of a parameter~$\param_i$ have been replaced by~$\pval(\param_i)$.
\LongVersion{%
	We denote as a \emph{timed automaton} any structure $\valuate{\A}{\pval}$, by assuming a rescaling of the constants: by multiplying all constants in $\valuate{\A}{\pval}$ by the least common multiple of their denominators, we obtain an equivalent (integer-valued) TA\LongVersion{, as defined in~\cite{AD94}}.
}

\LongVersion{
\subsubsection{Synchronized product of PTAs}
}

The \emph{synchronous product} (using strong broadcast, \ie{} synchronization on a given set of actions)\LongVersion{, or \emph{parallel composition},} of several PTAs gives a PTA\ea{linked to arXiv}.

\newcommand{\defSynchronizedProduct}{%

\begin{definition}[synchronized product of PTAs]\label{definition:parallel}
	Let $N \in \grandn$.
	Given a set of PTAs $\A_i = (\Actions_i, \Loc_i, (\locinit)_i, \Clock_i, \Param_i, \invariant_i, \Edges_i)$, $1 \leq i \leq N$,
	and a set of actions $\Actions_s$,
	the \emph{synchronized product} of $\A_i$, $1 \leq i \leq N$,
	denoted by $\A_1 \parallel_{\Actions_s} \A_2 \parallel_{\Actions_s} \cdots \parallel_{\Actions_s} \A_N$,
	is the tuple
		$(\Actions, \Loc, \locinit,  \Clock, \Param, \invariant, \Edges)$, where:
	\begin{enumerate}
		\item $\Actions = \bigcup_{i=1}^N\Actions_i$,
		\item $\Loc = \prod_{i=1}^N \Loc_i$,
		\item $\locinit = ((\locinit)_1, \dots, (\locinit)_N)$,
		\item $\Clock = \bigcup_{1 \leq i \leq N} \Clock_i$,
		\item $\Param = \bigcup_{1 \leq i \leq N} \Param_i$,
		\item $\invariant((\loc_1, \dots, \loc_N)) = \bigwedge_{i = 1}^{N} \invariant_i(\loc_i)$ for all $(\loc_1, \dots, \loc_N) \in \Loc$,
	\end{enumerate}
	and $\Edges{}$ is defined as follows.
	For all $\action \in \Actions$,
	let $\ActionsIndices_\action$ be the subset of indices $i \in 1, \dots, N$
	such that $\action \in \Actions_i$.
	For all  $\action \in \Actions$,
	for all $(\loc_1, \dots, \loc_N) \in \Loc$,
	for all \mbox{$(\loc_1', \dots, \loc_N') \in \Loc$},
	$\big((\loc_1, \dots, \loc_N), \guard, \action, \resets, (\loc'_1, \dots, \loc'_N)\big) \in \Edges$
	if:
	\begin{itemize}
		\item if $\action \in \Actions_s$, then
		\begin{enumerate}
			\item for all $i \in \ActionsIndices_\action$, there exist $\guard_i, \resets_i$ such that $(\loc_i, \guard_i, \action, \resets_i, \loc_i') \in \Edges_i$, $\guard = \bigwedge_{i \in \ActionsIndices_\action} \guard_i$, $\resets = \bigcup_{i \in \ActionsIndices_\action}\resets_i$, and,
			\item for all $i \not\in \ActionsIndices_\action$, $\loc_i' = \loc_i$.
		\end{enumerate}
		\item otherwise (if $\action \notin \Actions_s$), then there exists $i \in \ActionsIndices_\action$ such that
		\begin{enumerate}
			\item there exist $\guard_i, \resets_i$ such that $(\loc_i, \guard_i, \action, \resets_i, \loc_i') \in \Edges_i$, $\guard = \guard_i$, $\resets = \resets_i$, and,
			\item for all $j \neq i$, $\loc_j' = \loc_j$.
		\end{enumerate}
	\end{itemize}
\end{definition}

That is, synchronization is only performed on~$\Actions_s$, and other actions are interleaved.

}

\LongVersion{\defSynchronizedProduct{}}

\LongVersion{
\subsubsection{Concrete semantics of TAs}

Let us now recall the concrete semantics of TA.
}

\begin{definition}[Semantics of a TA]
	Given a PTA $\A = (\Actions, \Loc, \locinit, \Clock, \Param, \invariant, \Edges)$,
	and a parameter valuation~\(\pval\),
	the semantics of $\valuate{\A}{\pval}$ is given by the timed transition system (TTS) $(\States, \sinit, \flecheRel)$, with
	\begin{itemize}
		\item $\States = \{ (\loc, \clockval) \in \Loc \times \Rgeqzero^\ClockCard \mid \clockval \models \valuate{\invariant(\loc)}{\pval} \}$, %
		\LongVersion{\item} $\sinit = (\locinit, \ClocksZero) $,
		\item  $\flecheRel$ consists of the discrete and (continuous) delay transition relations:
		\begin{ienumeration}
			\item discrete transitions: $(\loc,\clockval) \longueflecheRel{\edge} (\loc',\clockval')$, %
				if $(\loc, \clockval) , (\loc',\clockval') \in \States$, and there exists $\edge = (\loc,\guard,\action,\resets,\loc') \in \Edges$, such that $\clockval'= \reset{\clockval}{\resets}$, and $\clockval\models\pval(\guard$).
			\item delay transitions: $(\loc,\clockval) \longueflecheRel{d} (\loc, \clockval+d)$, with $d \in \Rgeqzero$, if $\forall d' \in [0, d], (\loc, \clockval+d') \in \States$.
		\end{ienumeration}
	\end{itemize}
\end{definition}

    Moreover we write $(\loc, \clockval)\longuefleche{(\edge, d)} (\loc',\clockval')$ for a combination of a delay and discrete transition if
		$\exists  \clockval'' :  (\loc,\clockval) \longueflecheRel{d} (\loc,\clockval'') \longueflecheRel{\edge} (\loc',\clockval')$.

Given a TA~$\valuate{\A}{\pval}$ with concrete semantics $(\States, \sinit, \flecheRel)$, we refer to the states of~$\States$ as the \emph{concrete states} of~$\valuate{\A}{\pval}$.
A \emph{run} of~$\valuate{\A}{\pval}$ is an alternating sequence of concrete states of $\valuate{\A}{\pval}$ and pairs of edges and delays starting from the initial state $\sinit$ of the form
$\state_0, (\edge_0, d_0), \state_1, \cdots$
with
$i = 0, 1, \dots$, $\edge_i \in \Edges$, $d_i \in \Rgeqzero$ and
	$\state_i \longuefleche{(\edge_i, d_i)} \state_{i+1}$.
The \emph{duration} of a finite run $\varrun : \state_0, (\edge_0, d_0), \state_1, \cdots, \state_i$ is $\duration(\varrun) = \sum_{0 \leq j \leq i-1} d_j$.
Given\LongVersion{ a state}~$\state=(\loc, \clockval)$, we say that $\state$ is reachable in~$\valuate{\A}{\pval}$ if $\state$ appears in a run of $\valuate{\A}{\pval}$.
By extension, we say that $\loc$ is reachable\LongVersion{; and by extension again, given a set~$\somelocs$ of locations, we say that $\somelocs$ is reachable if there exists $\loc \in \somelocs$ such that $\loc$ is reachable in~$\valuate{\A}{\pval}$}.
Given $\loc, \loc' \in \Loc$ and a run~$\varrun$, we say that $\loc$ is reachable on the way to~$\loc'$ in~$\varrun$ if $\varrun$ is of the form $(\loc_0), (\edge_0, d_0), \cdots, (\edge_n, d_n), \cdots (\edge_m, d_m) \cdots$ for some~$m,n \in \grandn$ such that $\loc_n = \loc$, $\loc_m = \loc'$ and $\forall 0 \leq i \leq n-1, \loc_i \neq \loc'$.
Conversely, $\loc$ is unreachable on the way to~$\loc'$ in~$\varrun$ if $\varrun$ is of the form $(\loc_0), (\edge_0, d_0), \cdots, (\edge_m, d_m) \cdots$ with $\loc_m = \loc'$ and $\forall 0 \leq i \leq m-1, \loc_i \neq \loc$.

\ea{question for me: what if $\loc'$ (the final location) is not terminal?}
\ea{note for myself: what about Zeno runs?}

\begin{example}
	Consider again the PTA~$\A$ in \cref{figure:example-PTA}, and let $\pval$ be such that $\pval(\param_1) = 1$ and $\pval(\param_2) = 2$.
	Consider the following run~$\varrun$ of $\valuate{\A}{\pval}$:
	$(\loc_0, \clock = 0) , (\edge_2, 1.4) , (\loc_2, \clock = 1.4) , (\edge_3, 1.3) , (\loc_1, \clock = 2.7)$, where
		$\edge_2$ is the edge from $\loc_0$ to~$\loc_2$ in \cref{figure:example-PTA},
		and
		$\edge_3$ is the edge from $\loc_2$ to~$\loc_1$.
		We write ``$\clock = 1.4$'' instead of ``$\clockval$ such that $\clockval(\clock) = 1.4$''.
	We have $\duration(\varrun) = 1.4 + 1.3 = 2.7$.
	In addition, $\loc_2$ is reachable on the way to~$\loc_1$ in~$\varrun$.
\end{example}
\newcommand{\appendixSymbolicSemantics}{

\subsection{Symbolic semantics}\label{ss:symbolic}

Let us now recall the symbolic semantics of PTAs (see \eg{} \cite{HRSV02}\ea{add back \cite{ACEF09} to final version (double-blind)}).

\paragraph{Constraints}
We first need to define operations on constraints.
A linear term over $\Clock \cup \Param$ is of the form $\sum_{1 \leq i \leq \ClockCard} \alpha_i \clock_i + \sum_{1 \leq j \leq \ParamCard} \beta_j \param_j + d$, with
	$\clock_i \in \Clock$,
	$\param_j \in \Param$,
	and
	$\alpha_i, \beta_j, d \in \grandz$.
A \emph{constraint}~$\C$ (\ie{} a convex polyhedron) over $\Clock \cup \Param$ is a conjunction of inequalities of the form $\lterm \compOp 0$, where $\lterm$ is a linear term.

Given a parameter valuation~$\pval$, $\valuate{\C}{\pval}$ denotes the constraint over~$\Clock$ obtained by replacing each parameter~$\param$ in~$\C$ with~$\pval(\param)$.
Likewise, given a clock valuation~$\clockval$, $\valuate{\valuate{\C}{\pval}}{\clockval}$ denotes the expression obtained by replacing each clock~$\clock$ in~$\valuate{\C}{\pval}$ with~$\clockval(\clock)$.
We say that %
$\pval$ \emph{satisfies}~$\C$,
denoted by $\pval \models \C$,
if the set of clock valuations satisfying~$\valuate{\C}{\pval}$ is nonempty.
Given a parameter valuation $\pval$ and a clock valuation $\clockval$, we denote by $\wv{\clockval}{\pval}$ the valuation over $\Clock\cup\Param$ such that
for all clocks $\clock$, $\valuate{\clock}{\wv{\clockval}{\pval}}=\valuate{\clock}{\clockval}$
and
for all parameters $\param$, $\valuate{\param}{\wv{\clockval}{\pval}}=\valuate{\param}{\pval}$.
We use the notation $\wv{\clockval}{\pval} \models \C$ to indicate that $\valuate{\valuate{\C}{\pval}}{\clockval}$ evaluates to true.
We say that $\C$ is \emph{satisfiable} if $\exists \clockval, \pval \text{ s.t.\ } \wv{\clockval}{\pval} \models \C$.

We define the \emph{time elapsing} of~$\C$, denoted by $\timelapse{\C}$, as the constraint over $\Clock$ and $\Param$ obtained from~$\C$ by delaying all clocks by an arbitrary amount of time.
That is,
\[\wv{\clockval'}{\pval} \models \timelapse{\C} \text{ iff } \exists \clockval : \Clock \to \grandrplus, \exists d \in \grandrplus \text { s.t. } \wv{\clockval}{\pval} \models \C \land \clockval' = \clockval + d \text{.}\]
Given $\resets \subseteq \Clock$, we define the \emph{reset} of~$\C$, denoted by $\reset{\C}{\resets}$, as the constraint obtained from~$\C$ by resetting the clocks in~$\resets$, and keeping the other clocks unchanged.
We denote by $\projectP{\C}$ the projection of~$\C$ onto~$\Param$, \ie{} obtained by eliminating the variables not in~$\Param$ (\eg{} using Fourier-Motzkin\LongVersion{~\cite{Schrijver86})}.

\begin{definition}[Symbolic state]
	A symbolic state is a pair $(\loc, \C)$ where $\loc \in \Loc$ is a location, and $\C$ its associated parametric zone.
\end{definition}
\begin{definition}[Symbolic semantics]\label{def:PTA:symbolic}
	Given a PTA $\A = (\Actions, \Loc, \locinit, 
	\Clock, \Param, \invariant, \Edges)$,
	the symbolic semantics of~$\A$ is the labeled transition system called \emph{parametric zone graph}
	$ \PZG = ( \Edges, \SymbState, \symbstateinit, \symbtrans )$, with
	\begin{itemize}
		\item $\SymbState = \{ (\loc, \C) \mid \C \subseteq \invariant(\loc) \}$, %
		\item $\symbstateinit = \big(\locinit, \timelapse{(\bigwedge_{1 \leq i\leq\ClockCard}\clock_i=0)} \land \invariant(\loc_0) \big)$,
				and
		\item $\big((\loc, \C), \edge, (\loc', \C')\big) \in \symbtrans $ if $\edge = (\loc,\guard,\action,\resets,\loc') \in \Edges$ and
			\[\C' = \timelapse{\big(\reset{(\C \land \guard)}{\resets}\land \invariant(\loc')\big )} \land \invariant(\loc')\]
			with $\C'$ satisfiable.
	\end{itemize}

\end{definition}

That is, in the parametric zone graph, nodes are symbolic states, and arcs are labeled by \emph{edges} of the original PTA.

If $\big((\loc, \C), \edge, (\loc', \C')\big) \in \symbtrans $, we write $\Succ(\symbstate, \edge) = (\loc', \C')$.
By extension, we write $\Succ(\symbstate)$ for $\cup_{\edge \in \Edges} \Succ(\symbstate, \edge)$.

\begin{example}
	Consider again the PTA~$\A$ in \cref{figure:example-PTA}.
	The parametric zone graph of~$\A$ is given in \cref{figure:example-PTA:PZG}, where
		$\edge_1$ is the edge from $\loc_0$ to~$\loc_1$ in \cref{figure:example-PTA},
		$\edge_2$ is the edge from $\loc_0$ to~$\loc_2$,
		and
		$\edge_3$ is the edge from $\loc_2$ to~$\loc_1$.
	In addition, the symbolic states are:
	
	\begin{tabular}{r @{ } l  @{ }c  @{ }l  @{ }l}
		$\symbstate_0 =($ & $\loc_0$ & $,$ & $0 \leq \clock \leq 3 \land \param_1 \geq 0 \land \param_2 \geq 0 $ & $)$\\
		$\symbstate_1 =($ & $\loc_1$ & $,$ & $\clock \geq \param_2 \land 0 \leq \param_2 \leq 3 \land \param_1 \geq 0 $ & $)$\\
		$\symbstate_2 =($ & $\loc_2$ & $,$ & $3 \geq \clock \geq \param_1 \land 0 \leq \param_1 \leq 3 \land \param_2 \geq 0 $ & $)$\\
		$\symbstate_3 =($ & $\loc_1$ & $,$ & $\clock \geq \param_1 \land 0 \leq \param_1 \leq 3 \land \param_2 \geq 0 $ & $)$
		.
	\end{tabular}
\end{example}
\begin{figure}[tb]
 
	\centering
	 \footnotesize

	\begin{tikzpicture}[scale=1, xscale=2.5, yscale=2.2, auto, ->, >=stealth']
 
		\node[location, initial] at (0, 0) (s0) {$\symbstate_0$};
 
		\node[location] at (1, 0) (s2) {$\symbstate_2$};
 
		\node[location] at (0, -.5) (s1) {$\symbstate_1$};
 
		\node[location] at (1, -.5) (s3) {$\symbstate_3$};
 
		\path (s0) edge node[align=center]{$\edge_1$} (s1);
		\path (s0) edge[] node[align=center]{$\edge_2$} (s2);
		\path (s2) edge[] node[align=center]{$\edge_3$} (s3);

	\end{tikzpicture}
	\caption{Parametric zone graph of \cref{figure:example-PTA}}
	\label{figure:example-PTA:PZG}

\end{figure}
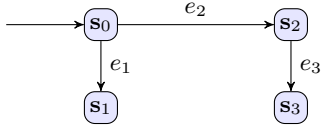

}
\LongVersion{
	\appendixSymbolicSemantics{}
}

\LongVersion{
\subsection{Reachability synthesis}
}

We will use reachability synthesis to solve the problems in \cref{section:problem}.
This procedure, called \EFsynth{}, takes as input a PTA~$\A$ and a set of target locations~$\somelocs$, and attempts to synthesize all parameter valuations~$\pval$ for which~$\somelocs$ is reachable in~$\valuate{\A}{\pval}$.
$\EFsynth(\A, \somelocs)$ was formalized in \eg{} \cite{JLR15} and is a procedure that may not terminate, but that computes an exact result (sound and complete) if it terminates.
\EFsynth{} traverses the \emph{parametric zone graph} of~$\A$\ShortVersion{, which is a potentially infinite extension of the well-known zone graph of TAs (see, \eg{} \cite{ACEF09,JLR15})}.
\begin{example}
	Consider again the PTA~$\A$ in \cref{figure:example-PTA}.
	$\EFsynth(\A, \{ \loc_1 \}) = \param_1 \leq 3 \lor \param_2 \leq 3$.
	Intuitively, it corresponds to all parameter constraints in the parametric zone graph \LongVersion{in \cref{figure:example-PTA:PZG} }associated to symbolic states with location~$\loc_1$.
\end{example}
\newcommand{\lemmaJLR}{

We \LongVersion{finally }recall the correctness of \EFsynth{}.

\begin{lemma}[\cite{JLR15}]\label{prop:EFsynth}
	Let $\A$ be a PTA, and let $\somelocs$ be a subset of the locations of~$\A$.
	Assume $\EFsynth(\A, \somelocs)$ terminates with result~$\K$.
	Then $\pval \models \K$ iff $\somelocs$ is reachable in~$\valuate{\A}{\pval}$.
\end{lemma}
}

\LongVersion{
	\lemmaJLR{}
}
\section{Timed-opacity problems}\label{section:problem}

Let us first introduce two key concepts to define our notion of opacity.
$\PrivDurReach{\valuate{\A}{\pval}}{\loc}{\loc'}$ (resp.\ $\PubDurReach{\valuate{\A}{\pval}}{\loc}{\loc'}$) is the set of the durations of the runs for which $\loc$ is reachable (resp.\ unreachable) on the way to~$\loc'$.
Formally:
$\PrivDurReach{\valuate{\A}{\pval}}{\loc}{\loc'} = \{ d \mid \exists \varrun $ in $ \valuate{\A}{\pval} $ such that $ d = \duration(\varrun) \land \loc$ is reachable on the way to~$\loc'$ in~$\varrun \}$
and
\(\PubDurReach{\valuate{\A}{\pval}}{\loc}{\loc'} = \{ d \mid \exists \varrun $ in $ \valuate{\A}{\pval} $ such that $ d = \duration(\varrun) \land \loc$ is unreachable on the way to~$\loc'$ in~$\varrun \}\).
\begin{example}\label{example:running:DurReach}
	Consider again the PTA in \cref{figure:example-PTA}, and let $\pval$ be such that $\pval(\param_1) = 1$ and $\pval(\param_2) = 2$.
	We have $\PrivDurReach{\valuate{\A}{\pval}}{\loc_2}{\loc_1} = [1, 3]$
	and
	$\PubDurReach{\valuate{\A}{\pval}}{\loc_2}{\loc_1} = [2, 3]$.
\end{example}

\ea{perhaps I should forbid outgoing transitions from~$\locfinal$ to make things simpler}

\ea{justify the use of the word execution time}

\begin{definition}[timed opacity \wrt{} $\Times$]\label{definition:ET-timed-opacity}
	Given a TA~$\valuate{\A}{\pval}$, a private location~$\locpriv$,
		a target location~$\locfinal$
		and a set of execution times~$\Times$,
		we say that $\valuate{\A}{\pval}$ is \emph{opaque \wrt{} $\locpriv$ on the way to~$\locfinal$ for execution times~$\Times$}
		if $\Times \subseteq \PrivDurReach{\valuate{\A}{\pval}}{\locpriv}{\locfinal} \cap \PubDurReach{\valuate{\A}{\pval}}{\locpriv}{\locfinal}$.
\end{definition}

\ea{NOTE!! modified the value last minute to remove the negative epsilon problem}
\begin{figure*}[tb]
\newcommand{\ratio}{0.5\textwidth}
 
	\centering
	 \footnotesize

\scalebox{.8}{

	\begin{tikzpicture}[scale=1, xscale=2.5, yscale=2.2, auto, ->, >=stealth']
 
		\node[location, initial] at (0,0) (l1) {$\loc_1$};
 
		\node[location] at (1, 0) (l2) {$\loc_2$};
 
		\node[location] at (2, 0) (l3) {$\loc_3$};
 
		\node[location] at (2, 1) (final1) {\styleloc{error}};
 
		\node[location] at (3, 0) (l4) {$\loc_4$};
 
		\node[location, private] at (4, 0) (hidden1) {$\locpriv$};
 
		\node[location] at (4, 1) (hidden2) {$\loc_5$};
 
		\node[location, final] at (5, 0) (final2) {$\locfinal$};
 
		\node[invariant, below=of l1] {$\styleclock{cl} \leq \styleparam{\epsilon}$};
		\node[invariant, below=of l2] {$\styleclock{cl} \leq \styleparam{\epsilon}$};
		\node[invariant, below=of l3] {$\styleclock{cl} \leq \styleparam{\epsilon}$};
		\node[invariant, below=of l4] {$\styleclock{cl} \leq \styleparam{\epsilon}$};
		
		\path (l1) edge node[align=center]{$\styleclock{cl} \leq \styleparam{\epsilon}$ \\ $\styleact{setupserver}$} node[below] {$\styleclock{cl} := 0$} (l2);
 
		\path (l2) edge node[align=center]{$\styleclock{cl} \leq \styleparam{\epsilon}$ \\  $\styleact{read?}\styledisc{x}$} node[below] {$\styleclock{cl} := 0$} (l3);

		\path (l3) edge node[above left, align=center]{$\styleclock{cl} \leq \styleparam{\epsilon}$ \\  $ \styledisc{x} < 0$} (final1);

		\path (l3) edge node[align=center]{$\styleclock{cl} \leq \styleparam{\epsilon}$ \\  $ \styledisc{x} \geq 0$ } node[below] {$\styleclock{cl} := 0$} (l4);

		\path (l4) edge node{\begin{tabular}{@{} c @{\ } c@{} }
		& $ \styledisc{x} \leq \styledisc{secret}$\\
		 $\land $ & $ \styleclock{cl} \leq \styleparam{\epsilon}$\\
		\end{tabular}} node [below, align=center]{$\styleclock{cl}:=0$} (hidden1);

		\path (l4) edge[bend left] node{\begin{tabular}{@{} c @{\ } c@{} }
		& $ \styledisc{x} > \styledisc{secret}$\\
		 $\land $ & $ \styleclock{cl} \leq \styleparam{\epsilon}$\\
		 & $\styleclock{cl}:=0$\\
		\end{tabular}} (hidden2);

		\path (hidden1) edge node{\begin{tabular}{@{} c @{\ } c@{} }
		& $ 32^2  \leq \styleclock{cl}$\\
		$\land$ & $ \styleclock{cl} \leq 32^2 + \styleparam{\epsilon}$\\
		\end{tabular}} (final2);

		\path (hidden2) edge[bend left] node{\begin{tabular}{@{} c @{\ } c@{} }
		& $ \styleparam{p} \times 32^2  \leq \styleclock{cl}$\\
		$\land$ & $ \styleclock{cl} \leq \styleparam{p} \times 32^2 + \styleparam{\epsilon}$\\
		\end{tabular}} (final2);
 
	\end{tikzpicture}
	
	}
	
	\caption{A Java program encoded in a PTA}
	\label{figure:example-Java:PTA}

\end{figure*}
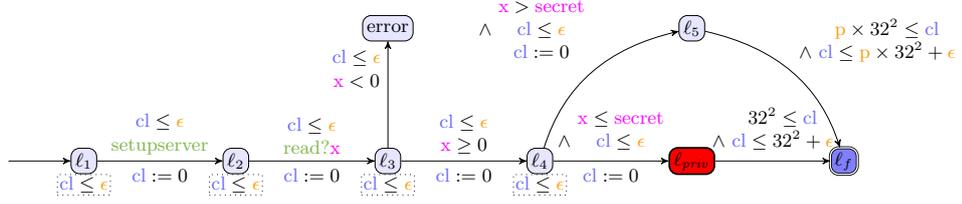

\begin{example}\label{example:Java-PTA}
	Consider the PTA~$\A$ in \cref{figure:example-Java:PTA} where $\styleclock{cl}$ is a clock, while $\styleparam{\epsilon},\styleparam{p}$ are parameters. %
	We use a sightly extended PTA syntax:
	$\styleact{read?}\styledisc{x}$ reads the value input on a given channel~$\styleact{read}$, and assigns it to a (discrete, global) variable~$\styledisc{x}$.
	$\styledisc{secret}$ is a constant variable of arbitrary value.
	If both $\styledisc{x}$ and $\styledisc{secret}$ are finite-domain variables (\eg{} bounded integers) then they can be seen as syntactic sugar for locations.
	Such variables are supported by most model checkers, including \uppaal{} and \imitator{}.
	
	This PTA encodes a server process
	from the DARPA Space/Time Analysis for Cybersecurity (STAC) library\LongVersion{\footnote{%
		\url{https://github.com/Apogee-Research/STAC/blob/master/Canonical_Examples/Source/Category1_vulnerable.java}
	}},
	that compares a user-input variable with a given secret
	and performs different actions taking different times depending on this secret.\ea{link to arXiv}
	\LongVersion{The Java code is given in \cref{appendix:Java}.
	}In our encoding, a single instruction takes a time in $[0, \styleparam{\epsilon}]$, while $\styleparam{p}$ is a (parametric) factor to one of the \stylecode{sleep} instructions of the program (originally, $\pval(\styleparam{p}) = 2$).
	For sake of simplicity, we abstract away instructions not related to time, and merge subfunctions calls.
	
	Fix $\pval(\styleparam{\epsilon}) = 1$, $\pval(\styleparam{p}) = 2$%
	.
	For this example,
		$\PrivDurReach{\valuate{\A}{\pval}}{\locpriv}{\locfinal} = [1024, 1029]$
		while
		$\PubDurReach{\valuate{\A}{\pval}}{\locpriv}{\locfinal} = [2048, 2053]$.
	Therefore, $\valuate{\A}{\pval}$ is opaque \wrt{} $\locpriv$ on the way to~$\locfinal$ for execution times~$\Times = [1024, 1029] \cap [2048, 2053] = \emptyset$.

	Now fix $\pval(\styleparam{\epsilon}) = 2$,	$\pval(\styleparam{p}) = 1.002$.
		$\PrivDurReach{\valuate{\A}{\pval}}{\locpriv}{\locfinal} = [1024, 1034]$
		while
		$\PubDurReach{\valuate{\A}{\pval}}{\locpriv}{\locfinal} = [1026.048, 1036.048]$.
	Therefore, $\valuate{\A}{\pval}$ is opaque \wrt{} $\locpriv$ on the way to~$\locfinal$ for execution times~$\Times = [1026.048, 1034]$.
	
\end{example}

\LongVersion{%
\subsection{Computation problems}
}

We can now define the timed-opacity computation problem, which consists in computing the possible execution times ensuring opacity \wrt{} a private location.
In other words, the attacker model is as follows: the attacker has only access to the computation time between the start of the program and the time it reaches a given (final) location.

\smallskip

\defProblem
	{Timed-opacity Computation}
	{A TA~$\valuate{\A}{\pval}$, a private location~$\locpriv$,
		a target location~$\locfinal$
		}
	{Compute the execution times~$\Times$ for which $\valuate{\A}{\pval}$ is opaque \wrt{} $\locpriv$ on the way to~$\locfinal$ for execution times~$\Times$}

\medskip

The synthesis counterpart allows for a higher-level problem by also synthesizing the internal timings guaranteeing opacity.

\smallskip

\defProblem
	{Timed-opacity Synthesis}
	{A PTA~$\A$, a private location~$\locpriv$,
		a target location~$\locfinal$
		}
	{Synthesize the parameter valuations~$\pval$ and the execution times~$\Times$ for which $\valuate{\A}{\pval}$ is opaque \wrt{} $\locpriv$ on the way to~$\locfinal$ for execution times~$\Times$}

Note that the execution times can depend on the parameter valuations.

\todo{example based on \cref{figure:example-PTA}?}

\section{Timed-opacity computation for timed automata}\label{section:TA}
\subsection{Answering the timed-opacity computation problem}

\todo{compare with existing results on non-interference (undecidable?)}

\ea{problem: control is sometimes decidable in~\cite{GMR07} but they control high-level only, and we'd like to control the only low-level final action (reaching the target location).}

\ea{Sun Jun: technically, I'd like to show that it's PSPACE-complete, but if I write so, I guess we would have to show that it's PSPACE-hard, and I have no idea how to do so…
I'm not really experienced with these decidability/complexity issues, so if you seen anything wrong, please let me know!
}

\begin{proposition}[timed-opacity computation]\label{proposition:ET-opacity-computation}
	The timed-opacity computation problem is solvable for TAs.
\end{proposition}
\newcommand{\proofETopacityComputation}{
\begin{proof}
	\ea{OK, I got it: we split the final location into final-via-priv and final-via-pub.
	We add an absolute clock (note: it won't ``diverge'' thanks to the classical extrapolation).
	We compute all values for this absolute clock reaching each final location.
	We take the intersection: it gives us the secure computation times.
	}%
	Let~$\A$ be a TA.
	We aim at exhibiting the execution times~$\Times$ for which $\A$ is opaque \wrt{} $\locpriv$ on the way to~$\locfinal$ for~$\Times$.
	We show in the following that this can be obtained from the region graph, the construction of which is EXPSPACE for timed automata~\cite{AD94}.
	
	We modify the TA as follows.
	First, let us add a new clock $\clockabs$, which is never reset in the TA.
	Second, we add a Boolean discrete variable~$\bflag$, initially $\BFalse$.
	Recall that discrete variables over a finite domain are syntactic sugar for \emph{locations}: therefore, $\locfinal$ with $\bflag = \BFalse$ and $\locfinal$ with $\bflag = \BTrue$ can be seen as two different locations.
	Then, we set $\bflag := \BTrue$ on any transition leading to~$\locpriv$; therefore, $\bflag = \BTrue$ denotes that $\locpriv$ has been visited.
	
	We can now compute $\PubDurReach{\valuate{\A}{\pval}}{\locpriv}{\locfinal}$ and $\PrivDurReach{\valuate{\A}{\pval}}{\locpriv}{\locfinal}$ from the region graph of this modified TA, \ie{} values of~$\clockabs$ reaching $\locfinal$ with $\bflag = \BFalse$ (resp.\ $\bflag = \BTrue$), as follows.
	For each region the discrete part of which is $\locfinal$ with $\bflag = \BFalse$ (resp.\ $\bflag = \BTrue$), gather the clock constraints; they come in the form of an integer part, and constraints on the fractional parts of the form $ \clock_i - \clock_j \sim c $ or $ \clock_i \sim c $ or $ - \clock_i \sim c $ for ${\sim} \in \{<, \leq\}$ and $c \in \grandn$.
	We then apply variable elimination by existential quantification to keep only constraints over~$\clockabs$ and obtain the set of (integer and fractional) valuations of~$\clockabs$ such that $\locfinal$ with $\bflag = \BFalse$ (resp.\ $\bflag = \BTrue$) is reachable.
	Recall that, from the region graph semantics~\cite{AD94}, the integer part can take a finite number of values, thanks to the use of an extrapolation (in its simplest form, all integer values below~$k$ must be enumerated, while values above $k$ are in the same equivalence class---where $k$ is the largest integer constant of the TA).
	This gives a finite graph, and therefore the values of $\clockabs$ reaching $\locfinal$ with $\bflag = \BFalse$ (resp.\ $\bflag = \BTrue$) can be represented as a finite set of (possibly punctual) intervals.

	After computing $\PubDurReach{\valuate{\A}{\pval}}{\locpriv}{\locfinal}$ and $\PrivDurReach{\valuate{\A}{\pval}}{\locpriv}{\locfinal}$, we can directly apply \cref{definition:ET-timed-opacity}:
	we perform the intersection of the valuations of~$\clockabs$ for which $\locfinal$ with $\bflag = \BFalse$ is reachable together with these for which $\locfinal$ with $\bflag = \BTrue$ is reachable, which gives the maximum set of execution times~$\Times$ for which $\valuate{\A}{\pval}$ is opaque \wrt{} $\locpriv$ on the way to~$\locfinal$ for execution times~$\Times$.

	Finally note that, while correct in theory, our construction could be largely improved in practice with the zone graph construction with appropriate extrapolations\LongVersion{ (\eg{} \cite{BBLP06,HSW16})}.
	Also note that our practical method will be different from this proof, reducing to reachability in a parametric model.
\end{proof}
}
\LongVersion{\proofETopacityComputation{}}

This positive result can be put in perspective with the negative result of~\cite{Cassez09}, that proves that it is undecidable whether a TA (and even the more restricted subclass of event-recording automata\LongVersion{~\cite{AFH99}}) is opaque, in a sense that the attacker can deduce some actions, by looking at observable actions together with their timing.
The difference in our setting is that only the global time is observable, which can be seen as a single action, occurring once only at the end of the computation.
In other words, our attacker is less powerful than the attacker in~\cite{Cassez09}.

\subsection{Checking for timed-opacity}

If one does not have the ability to tune the system (\ie{} change internal delays, or add some \stylecode{sleep()} or \stylecode{Wait()} statements in the program), one may be first interested in knowing whether the system is opaque for all execution times.

\begin{definition}[timed opacity]\label{definition:opacity}
	Given a TA~$\valuate{\A}{\pval}$, a private location~$\locpriv$ and
		a target location~$\locfinal$,
		we say that $\valuate{\A}{\pval}$ is \emph{opaque \wrt{} $\locpriv$ on the way to~$\locfinal$}
		if $\PrivDurReach{\valuate{\A}{\pval}}{\locpriv}{\locfinal} = \PubDurReach{\valuate{\A}{\pval}}{\locpriv}{\locfinal}$.
\end{definition}

That is, a system is opaque if, for any execution time~$d$, a run of duration~$d$ reaches~$\locfinal$ after passing by~$\locpriv$ iff another run of duration~$d$ reaches~$\locfinal$ without passing by~$\locpriv$.

\begin{remark}\label{remark:R2}
	This definition is symmetric: a system is not opaque iff an attacker can deduce $\locpriv$ or $\neg \locpriv$.
	For instance, if there is no path through $\locpriv$ to $\locfinal$, but a path to $\locfinal$, a system is not opaque \wrt{} \cref{definition:opacity}.
\end{remark}

As we have a procedure to compute $\PrivDurReach{\valuate{\A}{\pval}}{\locpriv}{\locfinal}$ and $\PubDurReach{\valuate{\A}{\pval}}{\locpriv}{\locfinal}$, (see \cref{proposition:ET-opacity-computation}),
\cref{definition:opacity} gives an immediate procedure to decide timed opacity.
Note that, from the finiteness of the region graph, $\PrivDurReach{\valuate{\A}{\pval}}{\locpriv}{\locfinal}$ and $\PubDurReach{\valuate{\A}{\pval}}{\locpriv}{\locfinal}$ come in the form of a finite union of intervals, and their equality can be effectively computed.

\begin{example}
	Consider again the PTA~$\A$ in \cref{figure:example-PTA}, and let $\pval$ be such that $\pval(\param_1) = 1$ and $\pval(\param_2) = 2$.
	Recall from \cref{example:running:DurReach} that
		$\PrivDurReach{\valuate{\A}{\pval}}{\loc_2}{\loc_1} = [1, 3]$
	and
	$\PubDurReach{\valuate{\A}{\pval}}{\loc_2}{\loc_1} = [2, 3]$.
	Thus, $\PrivDurReach{\valuate{\A}{\pval}}{\loc_2}{\loc_1} \neq \PubDurReach{\valuate{\A}{\pval}}{\loc_2}{\loc_1}$ and therefore $\valuate{\A}{\pval}$ is \emph{not} opaque \wrt{} $\loc_2$ on the way to~$\loc_1$.
	
	Now, consider $\pval'$ such that $\pval'(\param_1) = \pval'(\param_2) = 1.5$.
	This time, $\PrivDurReach{\valuate{\A}{\pval'}}{\loc_2}{\loc_1} = \PubDurReach{\valuate{\A}{\pval'}}{\loc_2}{\loc_1} = [1.5, 3]$ and therefore $\valuate{\A}{\pval'}$ is opaque \wrt{} $\loc_2$ on the way to~$\loc_1$.
\end{example}
\section{Decidability and undecidability}\label{section:theory}

We address here the following decision problem, that asks about the emptiness of the parameter valuations and execution times set guaranteeing timed opacity.

\smallskip

\defProblem
	{Timed-opacity Emptiness}
	{A PTA~$\A$, a private location~$\locpriv$,
		a target location~$\locfinal$
		}
	{Is the set of valuations~$\pval$ such that $\valuate{\A}{\pval}$ is opaque \wrt{} $\locpriv$ on the way to~$\locfinal$ for a non-empty set of execution times empty?}

Dually, we are interested in deciding whether there exists at least one parameter valuation for which $\valuate{\A}{\pval}$ is opaque for at least some execution time.

\LongVersion{
\subsection{Undecidability in general}
}

With the rule of thumb that all non-trivial decision problems are undecidable for general PTAs~\cite{Andre19STTT}, the following result is not surprising, and follows from the undecidability of reachability-emptiness for PTAs.

\begin{proposition}[undecidability]\label{proposition:undecidability}
	The timed-opacity emptiness problem is undecidable for general PTAs.
\end{proposition}
\newcommand{\proofUndecidability}{
\begin{proof}
	We reduce from the reachability-emptiness problem, \ie{} the existence of a parameter valuation reaching a given location in a PTA, which is undecidable~\ShortVersion{\cite{AHV93}}\LongVersion{\cite{AHV93,Miller00,Doyen07,JLR15,BBLS15}}.
	Consider an arbitrary PTA~$\A$ with initial location~$\locinit$ and a given location~$\locfinal$.
	It is undecidable whether there exists a parameter valuation for which there exists a run reaching~$\locfinal$.
	(Proofs of undecidability in the literature generally reduce from the halting problem of a 2-counter machine, so one can see~$\A$ as an encoding of a 2-counter machine.)
	Now, add the following locations and transitions (all unguarded) as in \cref{figure:undecidability}:
		a new urgent\footnote{%
			Where time cannot elapse (depicted in dotted yellow in our figures).
		} initial location~$\locinit'$ with outgoing transitions to~$\locinit$ %
			and to a new location~$\locpub$;
		a new urgent location $\locpriv$ with an %
			incoming transition from~$\locfinal$; %
		a new final location~$\locfinal'$ with incoming transitions from~$\locpriv$ %
			and~$\locpub$. %
	Also, $\locfinal$ is made urgent.
	Let~$\A'$ denote this new PTA.

	First note that, due to the unguarded transitions, $\locfinal'$ is reachable for any parameter valuation and for any execution time by runs passing by~$\locpub$ and not passing by~$\locpriv$.
	That is, for all~$\pval$, $\PubDurReach{\valuate{\A'}{\pval}}{\locpriv}{\locfinal'} = [0, \infty)$.
	
	Assume there exists some parameter valuation~$\pval$ such that $\locfinal$ is reachable from~$\locinit$ in~$\valuate{\A}{\pval}$ for some execution times~$\Times$: then, due to our construction with additional urgent locations, $\locpriv$ is reachable on the way to $\locfinal'$ in~$\valuate{\A'}{\pval}$ for the exact same execution times~$\Times$.
	Therefore, $\valuate{\A}{\pval}$ is opaque \wrt{} $\locpriv$ on the way to~$\locfinal'$ for execution times~$\Times$.
	
	Conversely, if $\locfinal$ is not reachable from~$\locinit$ in~$\A$ for any valuation, then $\locpriv$ is not reachable on the way to $\locfinal'$ for any valuation in~$\A'$.
	Therefore, there is no valuation~$\pval$ such that $\valuate{\A}{\pval}$ is opaque \wrt{} $\locpriv$ on the way to~$\locfinal'$ for any execution time.
	Therefore, there exists a valuation~$\pval$ such that $\valuate{\A}{\pval}$ is opaque \wrt{} $\locpriv$ on the way to~$\locfinal'$ iff $\locfinal$ is reachable in~$\A$---which is undecidable.
\end{proof}
}
\LongVersion{\proofUndecidability}
\begin{figure}[tb]
	{\centering
	\begin{tikzpicture}[->, >=stealth', auto, node distance=2cm, thin]

		\node[location] (l0) at (-2, 0) {$\locinit$};
		\node[location, urgent] (lf) at (+1.8, 0) {$\locfinal$};
		\node[cloud, cloud puffs=15.7, cloud ignores aspect, minimum width=5cm, minimum height=2cm, align=center, draw] (cloud) at (0cm, 0cm) {$\A$};

		\node[location, urgent, initial] (l0') at (-3.5, 0) {$\locinit'$};
		\node[location, urgent, private] (lpriv) at (+3.5, 0) {$\locpriv$};
		\node[location] (lpub) at (+2, -1.2) {$\locpub$};
		\node[location, final] (lf') at (+3.5, -1.2) {$\locfinal'$};

		\path
			(l0') edge (l0) %
			(lf) edge (lpriv) %
			(lpriv) edge (lf') %
			(l0') edge[out=-45,in=180] (lpub)
			(lpub) edge (lf')
			;

	\end{tikzpicture}

	}
	\caption{Reduction from reachability-emptiness} %
	\label{figure:undecidability}
\end{figure}
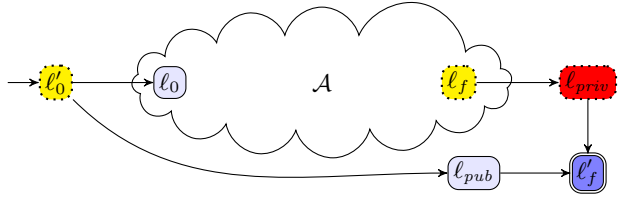

\LongVersion{
\subsection{A decidable subclass}
}

We now show that the timed-opacity emptiness problem is decidable for the subclass of PTAs called L/U-PTAs~\cite{HRSV02}.
Despite early positive results for L/U-PTAs\LongVersion{~\cite{HRSV02,BlT09}}, more recent results\LongVersion{ (notably \cite{JLR15,AM15,ALime17,ALR18FORMATS})} mostly proved undecidable properties of L/U-PTAs\ShortVersion{~\cite{Andre19STTT}}, and therefore this positive result is welcome.

\LongVersion{
	\paragraph{Syntax of L/U-PTAs}
}

\begin{definition}[L/U-PTA]\label{def:LUPTA} %
	An \emph{L/U-PTA} is a PTA where the set of parameters is partitioned into lower-bound parameters and upper-bound parameters,
	where each upper-bound (resp.\ lower-bound) parameter~$\param_i$ must be such that, 
    for every guard or invariant constraint $\clock \compOp \sum_{1 \leq i \leq \ParamCard} \alpha_i \param_i + d$, we have: $\alpha_i > 0$ implies ${\compOp} \in \{ \leq, < \}$ (resp.\ ${\compOp} \in \{ \geq, > \}$).
\end{definition}
\begin{example}
	The PTA in \cref{figure:example-PTA} is an L/U-PTA with $\{ \param_1, \param_2 \}$ as lower-bound parameters, and $\emptyset$ as upper-bound parameters.
	
	The PTA in \cref{figure:example-Java:PTA} is not an L/U-PTA, because $\styleparam{\param}$ is compared to $\styleclock{cl}$ both as a lower-bound (in ``$\styleparam{p} \times 32^2 \leq \styleclock{cl}$'') and as an upper-bound (``$\styleclock{cl} \leq \styleparam{p} \times 32^2 + \styleparam{\epsilon}$'').
\end{example}
\begin{theorem}[decidability]\label{proposition:decidability}
	The timed-opacity emptiness problem is decidable for L/U-PTAs.
\end{theorem}
\newcommand{\proofDecidability}{
\begin{proof}
	We reduce to the timed-opacity computation problem of a given TA, which is decidable (\cref{proposition:ET-opacity-computation}).

	Let $\A$ be an L/U-PTA.
	Let $\Azeroinf$ denote the structure obtained as follows: any occurrence of a lower-bound parameter is replaced with~0, and any occurrence of a conjunct $\clock \compOpLeq \param$ (where $\param$ is necessarily a upper-bound parameter) is deleted, \ie{} replaced with~$\CTrue$.

	Let us show that
		the set of valuations~$\pval$ such that $\valuate{\A}{\pval}$ is opaque \wrt{} $\locpriv$ on the way to~$\locfinal$ for a non-empty set of execution times is non empty
		iff
		the solution to the timed-opacity computation problem for~$\Azeroinf$ is non-empty.
	
	\begin{itemize}
		\item[$\Rightarrow$]
			Assume there exists a valuation~$\pval$ such that $\valuate{\A}{\pval}$ is opaque \wrt{} $\locpriv$ on the way to~$\locfinal$ for a non-empty set of execution.
			Therefore, the solution to the timed-opacity computation problem for~$\Azeroinf$ is non-empty.
			That is, there exists a duration~$d$ such that
				there exists a run of duration~$d$ such that $\locpriv$ is reachable on the way to~$\locfinal$,
			and
				there exists a run of duration~$d$ such that $\locpriv$ is unreachable on the way to~$\locfinal$.
			
			We now need the following monotonicity property of L/U-PTAs:

			\begin{lemma}[\cite{HRSV02}]\label{lemma:HRSV02:prop4.2}
				Let~$\A$ be an L/U-PTA and~$\pval$ be a parameter valuation.
				Let $\pval'$ be a valuation such that
				for each upper-bound parameter~$\param^+$, $\pval'(\param^+) \geq \pval(\param^+)$
				and
				for each lower-bound parameter~$\param^-$, $\pval'(\param^-) \leq \pval(\param^-)$.
				Then any run of~$\valuate{\A}{\pval}$ is a run of $\valuate{\A}{\pval'}$.
			\end{lemma}

			Therefore, from \cref{lemma:HRSV02:prop4.2}, the runs of~$\valuate{\A}{\pval}$ of duration~$d$ such that $\locpriv$ is reachable (resp.\ unreachable) on the way to~$\locfinal$
				are also runs of~$\Azeroinf$.
			Therefore, there exists a non-empty set of durations such that $\Azeroinf$ is opaque, \ie{} solution to the timed-opacity computation problem for~$\Azeroinf$ is non-empty.
			
		\item[$\Leftarrow$]
			Assume the solution to the timed-opacity computation problem for~$\Azeroinf$ is non-empty.
			That is, there exists a duration~$d$ such that
				there exists a run of duration~$d$ such that $\locpriv$ is reachable on the way to~$\locfinal$ in~$\Azeroinf$,
			and
				there exists a run of duration~$d$ such that $\locpriv$ is unreachable on the way to~$\locfinal$ in~$\Azeroinf$.
			
			The result could follow immediately---if only assigning $0$ and~$\infty$ to parameters was a proper parameter valuation.
			From~\LongVersion{\cite{HRSV02,BlT09}}\ShortVersion{\cite{HRSV02}}, if a location is reachable in the TA obtained by valuating lower-bound parameters with~0 and upper-bound parameters with~$\infty$, then there exists a sufficiently large constant~$C$ such that this run exists in $\valuate{\A}{\pval}$ such that $\pval$ assigns 0 to lower-bound and~$C$ to upper-bound parameters.
			Here, we can trivially pick~$d$, as any clock constraint $\clock \leq d$ will be satisfied for a run of duration~$d$.
			Let $\pval$ assign 0 to lower-bound and~$d$ to upper-bound parameters.
			Then,
				there exists a run of duration~$d$ such that $\locpriv$ is reachable on the way to~$\locfinal$ in~$\valuate{\A}{\pval}$,
			and
				there exists a run of duration~$d$ such that $\locpriv$ is unreachable on the way to~$\locfinal$ in~$\valuate{\A}{\pval}$.
			Therefore, the set of valuations~$\pval$ such that $\valuate{\A}{\pval}$ is opaque \wrt{} $\locpriv$ on the way to~$\locfinal$ for a non-empty set of execution times is non empty---which concludes the proof.
	\end{itemize}
\end{proof}
}
\LongVersion{\proofDecidability{}}
\begin{remark}
	The class of L/U-PTAs is known to be relatively meaningful, and many case studies from the literature fit into this class, including case studies proposed even before this class was defined in~\cite{HRSV02}.
	Even though the PTA in \cref{figure:example-Java:PTA} does not fit in this class, it can easily be transformed into an L/U-PTA, by duplicating $\styleparam{\param}$ into $\styleparam{\param^l}$ (used in lower-bound comparisons with clocks) and $\styleparam{\param^u}$ (used in upper-bound comparisons with clocks).
\end{remark}
\section{Parameter synthesis for opacity}\label{section:synthesis}

Despite the negative theoretical result of \cref{proposition:undecidability}, we now address the timed-opacity synthesis problem for the full class of PTAs.
Our method may not terminate (due to the undecidability) but, if it does, its result is correct.
Our workflow can be summarized as follows.

\begin{enumerate}
	\item We enrich the original PTA by adding a Boolean flag~$\bflag$ and a final synchronization action;
	\item We perform \emph{self-composition} (\ie{} parallel composition with a copy of itself) of this modified PTA;
	\item We perform reachability-synthesis using \EFsynth{} on~$\locfinal$ with contradictory values of~$\bflag$.
\end{enumerate}

We detail each operation in the following.
\LongVersion{

}In this section, we assume a PTA~$\A$, a given private location~$\locpriv$ and a given final location~$\locfinal$.

\subsubsection{Enriching the PTA}

We first add a Boolean flag $\bflag$ initially set to $\BFalse$, and then set to $\BTrue$ on any transition leading to $\locpriv$ (in the line of the proof of \cref{proposition:ET-opacity-computation}).
Therefore, $\bflag = \BTrue$ denotes that $\locpriv$ has been visited.
Second, we add a synchronization action $\actionEnd$ on any transition leading to~$\locfinal$.
Third, we add a new clock $\clockabs$ (never reset) together with a new parameter $\paramabs$, and we guard all transitions to~$\locfinal$ with $\clockabs = \paramabs$.
This will allow to measure the (parametric) execution time.
Let $\Enrich(\A, \locpriv, \locfinal)$ denote this procedure.

\LongVersion{

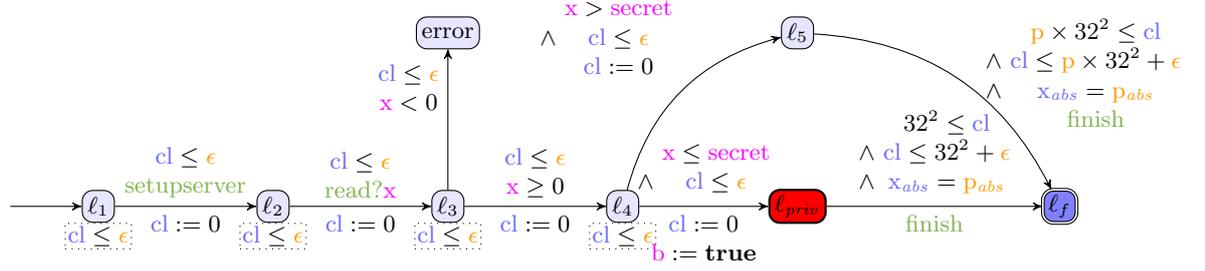
\begin{figure*}[tb]
\newcommand{\ratio}{0.5\textwidth}
 
	\centering
	\footnotesize

	\begin{tikzpicture}[scale=1, xscale=2.3, yscale=2.3, auto, ->, >=stealth']
 
		\node[location, initial] at (0,0) (l1) {$\loc_1$};
 
		\node[location] at (1, 0) (l2) {$\loc_2$};
 
		\node[location] at (2, 0) (l3) {$\loc_3$};
 
		\node[location] at (2, 1) (final1) {\styleloc{error}};
 
		\node[location] at (3, 0) (l4) {$\loc_4$};
 
		\node[location, private] at (4, 0) (hidden1) {$\locpriv$};
 
		\node[location] at (4, 1) (hidden2) {$\loc_5$};
 
		\node[location, final] at (5.5, 0) (final2) {$\locfinal$};
 
		\node[invariant, below=of l1] {$\styleclock{cl} \leq \styleparam{\epsilon}$};
		\node[invariant, below=of l2] {$\styleclock{cl} \leq \styleparam{\epsilon}$};
		\node[invariant, below=of l3] {$\styleclock{cl} \leq \styleparam{\epsilon}$};
		\node[invariant, below=of l4] {$\styleclock{cl} \leq \styleparam{\epsilon}$};
		
		\path (l1) edge node[align=center]{$\styleclock{cl} \leq \styleparam{\epsilon}$ \\ $\styleact{setupserver}$} node[below] {$\styleclock{cl} := 0$} (l2);
 
		\path (l2) edge node[align=center]{$\styleclock{cl} \leq \styleparam{\epsilon}$ \\  $\styleact{read?}\styledisc{x}$} node[below] {$\styleclock{cl} := 0$} (l3);

		\path (l3) edge node[above left, align=center]{$\styleclock{cl} \leq \styleparam{\epsilon}$ \\  $ \styledisc{x} < 0$} (final1);

		\path (l3) edge node[align=center]{$\styleclock{cl} \leq \styleparam{\epsilon}$ \\  $ \styledisc{x} \geq 0$ } node[below] {$\styleclock{cl} := 0$} (l4);

		\path (l4) edge node{\begin{tabular}{@{} c @{\ } c@{} }
		& $ \styledisc{x} \leq \styledisc{secret}$\\
		 $\land $ & $ \styleclock{cl} \leq \styleparam{\epsilon}$\\
		\end{tabular}} node [below, align=center]{$\styleclock{cl}:=0$ \\ $\styledisc{b}:=\CTrue$} (hidden1);

		\path (l4) edge[bend left] node{\begin{tabular}{@{} c @{\ } c@{} }
		& $ \styledisc{x} > \styledisc{secret}$\\
		 $\land $ & $ \styleclock{cl} \leq \styleparam{\epsilon}$\\
		 & $\styleclock{cl}:=0$\\
		\end{tabular}} (hidden2);

		\path (hidden1) edge node{\begin{tabular}{@{} c @{\ } c@{} }
		& $ 32^2  \leq \styleclock{cl}$\\
		$\land$ & $ \styleclock{cl} \leq 32^2 + \styleparam{\epsilon}$\\
		$\land$ & $\styleclock{ \clockabs }= \styleparam{\paramabs}$\\
		\end{tabular}} node[below]{$\actionEnd$} (final2);

		\path (hidden2) edge[bend left] node[right,xshift=8]{\begin{tabular}{@{} c @{\ } c@{} }
		& $ \styleparam{p} \times 32^2 \leq \styleclock{cl}$\\
		$\land$ & $ \styleclock{cl} \leq \styleparam{p} \times 32^2 + \styleparam{\epsilon}$\\
		$\land$ & $\styleclock{ \clockabs }= \styleparam{\paramabs}$\\
		& $\actionEnd$\\
		\end{tabular}} (final2);
 
	\end{tikzpicture}
	\caption{Transformed version of \cref{figure:example-Java:PTA}}
	\label{figure:example-Java:PTA-transformed}

\end{figure*}

\begin{example}
	\cref{figure:example-Java:PTA-transformed} shows the transformed version of the PTA in \cref{figure:example-Java:PTA}.
\end{example}
}

\subsubsection{Self-composition}

We use here the principle of \emph{self-composition}\ea{Sun Jun, any reference where this was discussed?}, \ie{} composing the PTA with a copy of itself.
More precisely, given a PTA~$\A' = \Enrich(\A, \locpriv, \locfinal)$, we first perform an identical copy of~$\A'$ \emph{with distinct variables}: that is, a clock~$\clock$ of~$\A'$ is distinct from a clock~$\clock$ in the copy of~$\A'$---which can be trivially performed using variable renaming.\footnote{%
	In fact, the fresh clock~$\clockabs$ and parameter~$\paramabs$ can be shared to save two variables, as~$\clockabs$ is never reset, and both PTAs enter~$\locfinal$ at the same time, therefore both ``copies'' of~$\clockabs$ and $\paramabs$ always share the same values.
}
Let~$\Copy(\A')$ denote this copy of~$\A'$.
We then compute $\A' \parallel_{\{\actionEnd\}} \Copy(\A')$.
That is, $\A'$ and~$\Copy(\A')$ evolve completely independently due to the interleaving---except that they are forced to enter~$\locfinal$ at the same time, thanks to the synchronization action~$\actionEnd$.

\subsubsection{Synthesis}

Then, we apply reachability synthesis \EFsynth{} (over all parameters, \ie{} the ``internal'' timing parameters, but also the $\paramabs$ parameter) to the following goal location:
the original~$\A'$ is in~$\locfinal$ with $\bflag = \BTrue$ while its copy~$\Copy(\A')$ is in~$\locfinal'$ with $\bflag' = \BFalse$ (primed variables denote variables from the copy).
Intuitively, we synthesize timing parameters and execution times such that there exists a run reaching~$\locfinal$ with $\bflag = \BTrue$ (\ie{} that has visited $\locpriv$) and there exists another run of same duration reaching~$\locfinal$ with $\bflag = \BFalse$ (\ie{} that has not visited $\locpriv$).

Let $\SynthOp(\A, \locpriv, \locfinal)$ denote the entire procedure.
We formalize $\SynthOp$ in \cref{algo:SynthOp}, where ``$\locfinal \land \bflag = \BTrue$'' denotes the location $\locfinal$ with $\bflag = \BTrue$.
Also note that \EFsynth{} is called on a set made of a single location of $\A' \parallel_{\{\actionEnd\}} \Copy(\A')$; by definition of the synchronous product, this location is a \emph{pair} of locations, one from~$\A'$ (\ie{} ``$\locfinal \land \bflag = \BTrue$'') and one from~$\Copy(\A')$ (\ie{} ``$\locfinal' \land \bflag' = \BFalse$'').

\begin{algorithm}[tb]
	\Input{A PTA $\A$, locations~$\locpriv,\locfinal$}
	\Output{Constraint $\K$ over the parameters}

	\LongVersion{\BlankLine}
	
	$\A' \assign \Enrich(\A, \locpriv, \locfinal)$
	
	$\A'' \assign \A' \parallel_{\{\actionEnd\}} \Copy(\A')$

	\Return $\EFsynth \Big(\A'', \big\{ ( \locfinal \land \bflag = \BTrue , \locfinal' \land \bflag' = \BFalse ) \big\} \Big)$
	
	\caption{$\SynthOp(\A, \locpriv, \locfinal)$}
	\label{algo:SynthOp}
\end{algorithm}

\ea{IMPORTANT NOTE: all the results are WRONG because we changed the mode (before $[sleep - \epsilon ; sleep + \epsilon]$ ; now $[sleep ; sleep + \epsilon]$}

\begin{example}
	Consider again the PTA~$\A$ in \cref{figure:example-Java:PTA}\LongVersion{: its enriched version~$\A'$ is given in \cref{figure:example-Java:PTA-transformed}}.
	Fix $\pval(\styleparam{\epsilon}) = 1$, $\pval(\styleparam{p}) = 2$.
	We then perform the synthesis applied to the self-composition of~$\A'$ according to \cref{algo:SynthOp}.
	The result obtained with \imitator{} is:
	$\paramabs = \emptyset$ (as expected from \cref{example:Java-PTA}).

	Now fix $\pval(\styleparam{\epsilon}) = 2$, $\pval(\styleparam{p}) = 1.002$.
	We obtain:
	$\paramabs \in [1026.048, 1034]$ (again, as expected from \cref{example:Java-PTA}).
	
	Now let us keep all parameters unconstrained.
	The result of \cref{algo:SynthOp} is the following 3-dimensional constraint:
	\ShortVersion{ $5 \times \styleparam{\epsilon} + 1024  \geq \styleparam{\paramabs} \geq 1024  \land 1024 \times \styleparam{p} + 5 \times \styleparam{\epsilon} \geq \styleparam{\paramabs} \geq 1024 \times \styleparam{p}  \geq 0$.
}

\LongVersion{
\begin{tabular}{l l}
 & $5 \times \styleparam{\epsilon} + 1024  \geq \styleparam{\paramabs} \geq 1024$
 \\
$ \land$ & $1024 \times \styleparam{p} + 5 \times \styleparam{\epsilon} \geq \styleparam{\paramabs} \geq 1024 \times \styleparam{p}  \geq 0$
 \end{tabular}
 }
\end{example}
\subsubsection{Soundness}

We will state below that, whenever $\SynthOp(\A, \locpriv, \locfinal)$ terminates, then its result is an exact (sound and complete) answer to the timed-opacity synthesis problem.

\newcommand{\proofsSoundness}{

Let us first prove a technical lemma used later to prove our the soundness of \SynthOp{}.

\begin{lemma}\label{lemma:ET}
	Assume $\SynthOp(\A, \locpriv, \locfinal)$ terminates with result~$\K$.
	For all $\pval \models \K$, there exists a run ending in~$\locfinal$ at time~$\pval(\paramabs)$ in~$\valuate{\A}{\pval}$.
\end{lemma}
\begin{proof}
	From the construction of~$\Enrich$, we added a new clock $\clockabs$ (never reset) together with a new parameter $\paramabs$, and we guarded all transitions to~$\locfinal$ with $\clockabs = \paramabs$.
	Therefore, valuations of~$\paramabs$ correspond exactly to the times at which $\locfinal$ can be reached in~$\valuate{\A}{\pval}$.
\end{proof}

We can now prove soundness and completeness.

\begin{proposition}[soundness]\label{proposition:soundness}
	Assume $\SynthOp(\A, \locpriv, \locfinal)$ terminates with result~$\K$.
	For all $\pval \models \K$, there exists a run of duration $\pval(\paramabs)$ such that $\locpriv$ is reachable on the way to~$\locfinal$ in~$\valuate{\A}{\pval}$
		and
	there exists a run of duration $\pval(\paramabs)$ such that $\locpriv$ is unreachable on the way to~$\locfinal$ in~$\valuate{\A}{\pval}$.
\end{proposition}
\begin{proof}
	$\SynthOp(\A, \locpriv, \locfinal)$ is the result of \EFsynth{} called on the self-composition of~$\Enrich(\A, \locpriv, \locfinal)$.
	Recall that $\Enrich$ has enriched~$\A$ with the addition of a guard $\clockabs = \paramabs$ on the incoming transitions of~$\locfinal$, as well as a Boolean flag $\bflag$ that is $\BTrue$ iff $\locpriv$ was visited along a run.
	Assume $\pval \models \K$.
	From \cref{prop:EFsynth}, there exists a run of $\A''$ reaching $ \locfinal \land \bflag = \BTrue , \locfinal' \land \bflag' = \BFalse$.
	From \cref{lemma:ET}, this run
	takes $\pval(\paramabs)$ time units.
	From the self-composition that is made of interleaving only (except for the final synchronization), there exists a run of duration $\pval(\paramabs)$ such that $\locpriv$ is reachable on the way to~$\locfinal$ in~$\valuate{\A}{\pval}$
	and
	there exists a run of duration $\pval(\paramabs)$ such that $\locpriv$ is unreachable on the way to~$\locfinal$ in~$\valuate{\A}{\pval}$.
\end{proof}
\begin{proposition}[completeness]\label{proposition:completeness}
	Assume $\SynthOp(\A, \locpriv, \locfinal)$ terminates with result~$\K$.
	Assume $\pval$.
	Assume there exists a run of duration $\pval(\paramabs)$ such that $\locpriv$ is reachable on the way to~$\locfinal$ in~$\valuate{\A}{\pval}$
		and
	there exists a run of duration $\pval(\paramabs)$ such that $\locpriv$ is unreachable on the way to~$\locfinal$ in~$\valuate{\A}{\pval}$.
	Then $\pval \models \K$.
\end{proposition}
\begin{proof}
	Assume $\SynthOp(\A, \locpriv, \locfinal)$ terminates with result~$\K$.
	Assume $\pval$.
	Assume there exists a run~$\varrun$ of duration $\pval(\paramabs)$ such that $\locpriv$ is reachable on the way to~$\locfinal$ in~$\valuate{\A}{\pval}$
		and
	there exists a run~$\varrun'$ of duration $\pval(\paramabs)$ such that $\locpriv$ is unreachable on the way to~$\locfinal$ in~$\valuate{\A}{\pval}$.
	
	First, from~$\Enrich$, there exists a run~$\varrun$ of duration $\pval(\paramabs)$ such that $\locpriv$ is reachable (resp.\ unreachable) on the way to~$\locfinal$ in~$\valuate{\A}{\pval}$ implies that there exists a run~$\varrun$ of duration $\pval(\paramabs)$ such that $\locfinal \land \bflag = \BTrue$ (resp.\ $\bflag = \BFalse$) is reachable in~$\valuate{\Enrich(\A)}{\pval}$.
	
	Since our self-composition allows any interleaving, runs~$\varrun$ of $\valuate{\A'}{\pval}$ and~$\varrun'$ in~$\valuate{\Copy(\A')}{\pval}$ are independent---except for reaching~$\locfinal$.
	Since $\varrun$ and $\varrun'$ have the same duration $\pval(\paramabs)$, then they both reach $\locfinal$ at the same time and, from our definition of self-composition, they can simultaneously fire action~$\actionEnd$ and enter~$\locfinal$ at time~$\pval(\paramabs)$.
	Hence, there exists a run reaching $\locfinal \land \bflag = \BTrue , \locfinal' \land \bflag' = \BFalse$ in~$\valuate{\A''}{\pval}$.
	
	Finally, from \cref{prop:EFsynth}, $\pval \models \K$.
\end{proof}

}

\LongVersion{
	\proofsSoundness{}
}
\begin{theorem}[correctness]\label{theorem:correctness}
	Assume $\SynthOp(\A, \locpriv, \locfinal)$ terminates with result~$\K$.
	Assume $\pval$.
	The following two statements are equivalent:
	\begin{enumerate}
		\item
			There exists a run of duration $\pval(\paramabs)$ such that $\locpriv$ is reachable on the way to~$\locfinal$ in~$\valuate{\A}{\pval}$
				and
			there exists a run of duration $\pval(\paramabs)$ such that $\locpriv$ is unreachable on the way to~$\locfinal$ in~$\valuate{\A}{\pval}$.
		\item $\pval \models \K$.
	\end{enumerate}
\end{theorem}
\LongVersion{
\begin{proof}
	From \cref{proposition:soundness,proposition:completeness}
\end{proof}
}

\todoinline{
\begin{example}
	Get the largest possible set of execution times?
	For \cref{figure:example-Java:PTA}, The answer shall be $\styleparam{p_1} = \styleparam{p_2} \land \styleparam{p} = 1$
\end{example}
}

\section{Experiments}\label{section:experiments}

\LongVersion{
\subsection{Experimental environment}
}

We use \imitator{}~\cite{AFKS12}, a tool taking as input networks of PTAs extended with several handful features such as shared global discrete variables, PTA synchronization through strong broadcast, etc.
\LongVersion{%
	\imitator{} represents symbolic states as polyhedra, relying on PPL~\cite{BHZ08}.

}%
We ran experiments using \imitator{} 2.10.4 ``Butter Jellyfish''\LongVersion{ (build 2477 \texttt{HEAD/5b53333})}
on a Dell XPS 13 9360 equipped with an Intel\textregistered{} Core\texttrademark{} i7-7500U CPU @ 2.70GHz with 8\,GiB memory running Linux Mint 18.3 64\,bits.\footnote{%
	Sources, models and results are available at
	\href{https://doi.org/10.5281/zenodo.3251141}{\nolinkurl{doi.org/10.5281/zenodo.3251141}}%
	\LongVersion{and
		\href{https://www.imitator.fr/static/ATVA19/}{\nolinkurl{imitator.fr/static/ATVA19/}}}.
}

\subsection{Translating programs into PTAs}\label{ss:Java2PTA}

We will consider case studies from the PTA community and from previous works focusing on privacy using (parametric) timed automata.
In addition, we will be interested in analyzing programs too.
In order to apply our method to the analysis of programs, we need a systematic way of translating a program (\eg{} a Java program) into a PTA.
In general, precisely modeling the execution time of a program using models like timed automata is highly non-trivial due to complication of hardware pipelining, caching, OS scheduling, etc.
The readers are referred to the rich literature in, for instance, \cite{LYGY10}.
In this work, we instead make the following simplistic assumption on execution time of a program statement and focus on solving the parameter synthesis problem.
How to precisely model the execution time of programs is orthogonal and complementary to our work.

We assume that the execution time of a program statement other than \stylecode{Thread.sleep(n)} is within a range $[0,\epsilon]$ where $\epsilon$ is a small integer constant (in milliseconds), whereas the execution time of statement \stylecode{Thread.sleep(n)} is within a range $[n , n+\epsilon]$.
In fact, we choose to keep $\epsilon$ \emph{parametric} to be as general as possible, and to not depend on particular architectures.

Our test subject is a set of benchmark programs from the DARPA Space/Time Analysis for Cybersecurity (STAC) program.\footnote{\url{https://github.com/Apogee-Research/STAC/}}
	These programs are being released publicly to facilitate researchers to develop methods and tools for identifying STAC vulnerabilities in the programs.

\todo{explain we merge ``wait'' loops into a static number of it (as you did). We also merge subfunction calls}

\ea{perhaps there is some literature on Java to TAs?}

\todo{check \cite[Fig3.a and b]{VNN18}; although not sure it'll terminate!! Also Fig.5 (\cref{figure:example:VNN18}) is perfect for us! (and will terminate)}

\todo{\cite[Fig.12]{BCLR15} is a simple yet interesting example for our framework}

\subsection{A richer framework}

The symbolic representation of variables and parameters in \imitator{} allows us to reason \emph{symbolically} concerning variables.
That is, instead of enumerating all possible (bounded) values of \styledisc{x} and \styledisc{secret} in \cref{figure:example-Java:PTA}, we turn them to parameters (\ie{} unknown constants), and \imitator{} performs a symbolic reasoning.
Even better, the analysis terminates for this example even when no bound is provided on these variables.
This is often not possible in (non-parametric) timed automata based model checkers, that usually have to enumerate these values.
Therefore, in our PTA representation of Java programs, we turn all user-input variable and secret constant variables to parameters.
Other local variables are implemented using \imitator{} discrete (shared, global) variables.

\LongVersion{We also discuss \ShortVersion{in \cref{appendix:richer} }how to enlarge the scope of our framework.}

\newcommand{\richerFramework}{
\paragraph{Multiple private locations}
This can be easily achieved by setting $\bflag$ to $\BTrue$ along any incoming transition of one of these private locations.
\ea{useful?}

\paragraph{Multiple final locations}
The technique used depends on whether these multiple final locations can be distinguished or not.
If they are indistinguishable (\ie{} the observer knows when the program has terminated, but not in which state), then it suffices to merge all these final locations in a single one, and our framework trivially applies.
If they are distinguishable, then one analysis needs to be conducted on each of these locations (with a different parameter~$\paramabs$ for each of these), and the obtained constraints must be intersected.\ea{useful?}

\paragraph{Access to high-level variables}
In the literature, a distinction is sometimes made between low-level (``public'') and high-level (``private'') variables.
Opacity or non-interference can be defined in terms of the ability for an observer to deduce some information on the high-level variables.

\begin{example}
	For example, in \cref{figure:example:VNN18} (where $\styleclock{cl}$ is a clock and $\styledisc{h}$ a variable), if $\loc_2$ is reachable in 20 time units, then it is clear that the value of the high-level variable $\styledisc{h}$ is negative.
\end{example}

Our framework can also be used to address this problem, \eg{} by setting $\bflag$ to~$\BTrue$, not on locations but on selected tests / assignments of such variables.

\begin{example}
	For example, setting $\bflag$ to~$\BTrue$ on the upper transition from~$\loc_1$ to~$\loc_2$ in \cref{figure:example:VNN18}, the answer to the timed-opacity computation problem is $\Times = (30, \infty)$, and the system is therefore not opaque since $\loc_2$ can be reached for any execution time in~$[0, \infty)$.
	\todo{new figure too}
\end{example}
\begin{figure}[tb]
	\centering
	\footnotesize
	\begin{tikzpicture}[scale=1, auto, ->, >=stealth']

		\node[location, initial] at (0,0) (l1) {$\loc_1$};

		\node[location, final] at (2, 0) (l2) {$\loc_2$};

		\path (l1) edge[bend left] node[align=center]{$\styledisc{h} > 0$ \\ $\styleclock{cl} > 30$} (l2);

		\path (l1) edge[bend right] node[below,align=center]{$\styledisc{h} \leq 0$} (l2);

	\end{tikzpicture}
	\caption{\cite[Fig.~5]{VNN18}}
	\label{figure:example:VNN18}

\end{figure}

}

\LongVersion{\richerFramework{}}

\ea{

measuring some time in some locations? quite easy probably
	Remark: measuring time in some locations starts to resemble language equivalence under partial observation!

But as this is quite exceeding our framework, I'd suggest to leave it for next time.
What do you think?
}

\subsection{Experiments}
\subsubsection{Benchmarks}

As a proof of concept, we applied our method to a set of examples from the literature.
The first five models come from previous works from the literature~\cite{GMR07,BCLR15,VNN18}, also addressing non-interference or opacity in timed automata.\LongVersion{\footnote{%
	As most previous works on opacity and timed automata do not come with an implementation nor with benchmarks, it is not easy to find larger models coming in the form of TAs.
}}
In addition, we used two common models from the (P)TA literature, not necessarily linked to security:
	a toy coffee machine (\stylebench{Coffee}) used as benchmark in a number of papers,
	and
	a model Fischer's mutual exclusion protocol (\stylebench{Fischer-HRSV02}) \cite{HRSV02}.
In both cases, we added manually a definition of private location (the number of sugars ordered, and the identity of the process entering the critical section, respectively), and we verified whether they are opaque \wrt{} these internal behaviors.

We also applied our approach to a set of Java programs from the aforementioned STAC library.
We use identifiers of the form \stylebench{STAC:1:n} where \stylebench{1} denotes the identifier in the library, while \stylebench{n} (resp.~\stylebench{v}) denotes non-vulnerable (resp.\ vulnerable).
We manually translated these programs to parametric timed automata, following the method described in \cref{ss:Java2PTA}.
We used a representative set of programs from the library; however, some of them were too complex to fit in our framework, notably when the timing leaks come from calls to external libraries (\stylebench{STAC:15:v}), when dealing with complex computations such as operations on matrices (\stylebench{STAC:16:v}) or when handling probabilities (\stylebench{STAC:18:v}).
Proposing efficient and accurate ways to represent arbitrary programs into (parametric) timed automata is orthogonal to our work, and is the object of future works.

\subsubsection{Timed-opacity computation}

First, we \emph{verified} whether a given TA model is opaque, \ie{} if for all execution times reaching a given final location, both an execution passes by a given private location and an execution does not pass by this private location.
To this end, we also answer the timed-opacity computation problem, \ie{} to synthesize all execution times for which the system is opaque.
While this problem can be verified on the region graph (\cref{proposition:ET-opacity-computation}), we use the same framework as in \cref{section:synthesis}, but without parameters in the original TA.
That is, we use the Boolean flag~$\bflag$ and the parameter $\paramabs$ to compute all possible execution times.
In other words, we use a parametric analysis to solve a non-parametric problem.

\newcommand{\columnRef}[1]{}

\begin{table}[tb]
	\caption{Experiments: timed opacity}
	\centering
	\footnotesize
	\ShortVersion{\scriptsize}
	
	\setlength{\tabcolsep}{2pt} %

	\LongVersionTable{\begin{tabular}{| c | c | c | c | c | c | r | r | c |}} %
	\ShortVersionTable{\begin{tabular}{| c | c | c | c | c | r | r | c |}} %
		\hline
		\LongVersionTable{\multicolumn{3}{| c |}{\cellHeader{}Model} & \multicolumn{3}{ c |}{\cellHeader{}Transf. PTA} & \multicolumn{3}{c |}{\cellHeader{}Result}}
		\ShortVersionTable{\multicolumn{3}{| c |}{\cellHeader{}Model} & \multicolumn{3}{ c |}{\cellHeader{}Transf. PTA} & \multicolumn{2}{c |}{\cellHeader{}Result}}
		\\
		\hline
		\rowHeader{}
		Name & \columnRef{Reference & }$|\A|$ & $|\Clock|$ & $|\A|$ & $|\Clock|$ & $|\Param|$\LongVersionTable{ & States} &Time (s) & Vulnerable? \\
		\hline
		\cite[Fig.~5]{VNN18} & \columnRef{\cite{VNN18} &} 1 & 1 & 2 & 3 & 3 \LongVersionTable{ & 13} & 0.02 & \cellFixable{}  \\ %
		\hline
		\cite[Fig.~1b]{GMR07} & \columnRef{\cite{GMR07} & }1 & 1 & 2 & 3 & 1 \LongVersionTable{& 25} & 0.04 & \cellFixable{} \\ %
		\hline
		\cite[Fig.~2a]{GMR07} & \columnRef{\cite{GMR07} & }1 & 1 & 2 & 3 & 1 \LongVersionTable{& 41} & 0.05 & \cellFixable{} \\ %
		\hline
		\cite[Fig.~2b]{GMR07} & \columnRef{\cite{GMR07} & }1 & 1 & 2 & 3 & 1 \LongVersionTable{& 41} & 0.02 & \cellFixable{} \\ %
		\hline
		Web privacy problem \cite{BCLR15} & \columnRef{\cite{BCLR15} & }1 & 2 & 2 & 4 & 1 \LongVersionTable{& 105} & 0.07 & \cellFixable{} \\ %
		\hline
		\stylebench{Coffee} & \columnRef{\cite{Andre18FTSCS} & }1 & 2 & 2 & 5 & 1 \LongVersionTable{& 43} & 0.05 & \cellNo{} \\ %
		\hline
		\stylebench{Fischer-HSRV02} & \columnRef{\cite{HRSV02} & }3 & 2 & 6 & 5 & 1 \LongVersionTable{& 2495} & 5.83 & \cellFixable{} \\ %
		\hline
		\stylebench{STAC:1:n} & \columnRef{ & }\multicolumn{2}{c |}{\nbLoC{69}} & 2 & 3 & 6 \LongVersionTable{& 65} & 0.12 & \cellFixable{} \\ %
		\hline
		\stylebench{STAC:1:v} & \columnRef{ & }\multicolumn{2}{c |}{\nbLoC{69}} & 2 & 3 & 6 \LongVersionTable{& 63} & 0.11 & \cellYes{} \\ %
		\hline
		\stylebench{STAC:3:n} & \columnRef{ & }\multicolumn{2}{c |}{\nbLoC{87}} & 2 & 3 & 8 \LongVersionTable{& 289} & 0.72 & \cellNo{} \\ %
		\hline
		\stylebench{STAC:3:v} & \columnRef{ & }\multicolumn{2}{c |}{\nbLoC{87}} & 2 & 3 & 8 \LongVersionTable{& 287} & 0.74 & \cellFixable{} \\ %
		\hline
		\stylebench{STAC:4:n} & \columnRef{ & }\multicolumn{2}{c |}{\nbLoC{112}} & 2 & 3 & 8 \LongVersionTable{& 904} & 6.40 & \cellYes{} \\ %
		\hline
		\stylebench{STAC:4:v} & \columnRef{ & }\multicolumn{2}{c |}{\nbLoC{110}} & 2 & 3 & 8 \LongVersionTable{& 19183} & 265.52 & \cellYes{} \\ %
		\hline
		\stylebench{STAC:5:n} & \columnRef{ & }\multicolumn{2}{c |}{\nbLoC{115}} & 2 & 3 & 6 \LongVersionTable{& 144} & 0.24 & \cellNo{} \\ %
		\hline
		\stylebench{STAC:11A:v} & \columnRef{ & }\multicolumn{2}{c |}{\nbLoC{81}} & 2 & 3 & 8 \LongVersionTable{& 5037} & 47.77 & \cellFixable{} \\ %
		\hline
		\stylebench{STAC:11B:v} & \columnRef{ & }\multicolumn{2}{c |}{\nbLoC{85}} & 2 & 3 & 8 \LongVersionTable{& 5486} & 59.35 & \cellFixable{} \\ %
		\hline
		\stylebench{STAC:12c:v} & \columnRef{ & }\multicolumn{2}{c |}{\nbLoC{81}} & 2 & 3 & 8 \LongVersionTable{& 1177} & 18.44 & \cellYes{} \\ %
		\hline
		\stylebench{STAC:12e:n} & \columnRef{ & }\multicolumn{2}{c |}{\nbLoC{96}} & 2 & 3 & 8 \LongVersionTable{& 169} & 0.58 & \cellYes{} \\ %
		\hline
		\stylebench{STAC:12e:v} & \columnRef{ & }\multicolumn{2}{c |}{\nbLoC{85}} & 2 & 3 & 8 \LongVersionTable{& 244} & 1.10 & \cellFixable{} \\ %
		\hline
		\stylebench{STAC:14:n} & \columnRef{ & }\multicolumn{2}{c |}{\nbLoC{88}} & 2 & 3 & 8 \LongVersionTable{& 1223} & 22.34 & \cellFixable{} \\ %
		\hline
	\end{tabular}
	
	\label{table:nonparametric}
\end{table}

We tabulate the experiments results in \cref{table:nonparametric}.
We give from left to right the model name, the numbers of automata and of clocks in the original timed automaton (this information is not relevant for Java programs as the original model is not a TA), the numbers of automata, of clocks and of parameters in the transformed PTA, the computation time in seconds (for the timed-opacity computation problem), and the result.
In the result column, ``\cellNo{}'' (resp.~``\cellYes{}'') denotes that the model is opaque (resp.\ is not opaque), while ``\cellFixable{}'' denotes that the model is not opaque, but could be fixed.
That is, although $\PrivDurReach{\valuate{\A}{\pval}}{\locpriv}{\locfinal} \neq \PubDurReach{\valuate{\A}{\pval}}{\locpriv}{\locfinal}$, their intersection is non-empty and therefore, by tuning the computation time, it may be possible to make the system opaque.
This will be discussed in \cref{ss:rendering-opaque}.

Even though we are interested here in timed opacity computation (and not in synthesis), note that all models derived from Java programs feature the parameter~$\styleparam{\epsilon}$.
The result is obtained by variable elimination, \ie{} by existential quantification over the parameters different from~$\paramabs$.
In addition, the number of parameters is increased by the parameters encoding the symbolic variables (such as $\styledisc{x}$ and $\styledisc{secret}$ in \cref{figure:example-Java:PTA}).

\paragraph*{Discussion}
Overall, our method is able to answer the timed-opacity computation problem relatively fast, exhibiting which execution times are opaque (timed-opacity computation problem), and whether \emph{all} execution times indeed guarantee opacity (timed-opacity problem).

In many cases, while the system is not opaque, we are able to \emph{infer} the execution times guaranteeing opacity (cells marked ``\cellFixable{}'').
This is an advantage of our method \wrt{} methods outputting only binary answers.

We observed some mismatches in the Java programs, \ie{} some programs marked \stylebench{n} (non-vulnerable) in the library are actually vulnerable according to our method.
This mainly comes from the fact that the STAC library uses some statistical analyses on the execution times, while we use an exact method.
Therefore, a very small mismatch between $\PrivDurReach{\valuate{\A}{\pval}}{\locpriv}{\locfinal}$ and $\PubDurReach{\valuate{\A}{\pval}}{\locpriv}{\locfinal}$ will lead our algorithm to answer ``not opaque'', while statistical methods may not be able to differentiate this mismatch from noise.
This is notably the case of \stylebench{STAC:14:n} where some action lasts either 5,010,000 or 5,000,000 time units depending on some secret, which our method detects to be different, while the library does not.
For \stylebench{STAC:1:n}, using our data, the difference in the execution time upper bound between an execution performing some secret action and an execution not performing it is larger than~1\,\%, which we believe is a value which is not negligible, and therefore this case study might be considered as vulnerable.
For \stylebench{STAC:4:n}, we used a different definition of opacity (whether the user has input the correct password, vs.\ information on the real password), which explains the mismatch.

Concerning the Java programs, we decided to keep the most abstract representation, by imposing that each instruction lasts for a time in~$[0,\styleparam{\epsilon}]$, with $\styleparam{\epsilon}$ a parameter.
However,
	fixing an identical (parametric) time $\styleparam{\epsilon}$ for all instructions,
	or fixing an arbitrary time in a constant interval $[0, \epsilon]$ (for some constant $\epsilon$, \eg{}~1),
	or even fixing an identical (constant) time $\epsilon$ (\eg{}~1) for all instructions,
significantly speeds up the analysis.
These choices can be made for larger models.
\subsubsection{Timed opacity synthesis}

Then, we address the timed-opacity synthesis problem.
In this case, we \emph{synthesize} both the execution time and the internal values of the parameters for which one cannot deduce private information from the execution time.

We consider the same case studies as for timed-opacity computation; however, the Java programs feature no internal ``parameter'' and cannot be used here.
Still, we artificially enriched one of them (\stylebench{STAC:3:v}) as follows: in addition to the parametric value of $\styleparam{\epsilon}$ and the execution time, we parameterized one of the \stylecode{sleep} timers.
The resulting constraint can help designers to refine this latter value to ensure opacity.

We tabulate the results in \cref{table:parametric}, where the columns are similar to \cref{table:nonparametric}.
A difference is that the first $|\Param|$ column denotes the number of parameters in the original model (without counting these added by our transformation).
In addition, \cref{table:parametric} does not contain a ``vulnerable?'' column as we \emph{synthesize} the condition for which the model is non-vulnerable, and therefore the answer is non-binary.
However, in the last column (``Constraint''), we make explicit whether no valuations ensure opacity (``\cellKnone{}''), all of them (``\cellKall{}''), or some of them (``\cellKsome{}'').

\paragraph*{Discussion}
An interesting outcome is that the computation time is comparable to the (non-parametric) timed-opacity computation, with an increase of up to~20\,\% only.
In addition, for all case studies, we exhibit at least some valuations for which the system can be made opaque.
Also note that our method always terminates for these models, and therefore the result exhibited is complete.
Interestingly, \stylebench{Coffee} is opaque for any valuation of the 3~internal parameters.

\begin{table}[tb]
	\caption{Experiments: timed opacity synthesis}
	\centering
	\footnotesize
	\ShortVersion{\scriptsize}
	
	\setlength{\tabcolsep}{2pt} %
	\LongVersionTable{\begin{tabular}{| c | c | c | c | c | c | c | r | r | c |}} %
	\ShortVersionTable{\begin{tabular}{| c | c | c | c | c | c | r | r | c |}} %
		\hline
		\LongVersionTable{\multicolumn{4}{| c |}{\cellHeader{}Model} & \multicolumn{3}{ c |}{\cellHeader{}Transf. PTA} & \multicolumn{3}{c |}{\cellHeader{}Result}}
		\ShortVersionTable{\multicolumn{4}{| c |}{\cellHeader{}Model} & \multicolumn{3}{ c |}{\cellHeader{}Transf. PTA} & \multicolumn{2}{c |}{\cellHeader{}Result}}
		\\
		\hline
		\rowHeader{}
		Name & \columnRef{Reference & }$|\A|$ & $|\Clock|$ & $|\Param|$ & $|\A|$ & $|\Clock|$ & $|\Param|$ \LongVersionTable{& States} & Time (s) & Constraint \\
		\hline
		\cite[Fig.~5]{VNN18} & \columnRef{\cite{VNN18} &} 1 & 1 & 0 & 2 & 3 & 4 \LongVersionTable{& 13} & 0.02 & \cellKsome{}\\
		\hline
		\cite[Fig.~1b]{GMR07} & \columnRef{\cite{GMR07} & }1 & 1 & 0 & 2 & 3 & 3 \LongVersionTable{& 25} & 0.03 & \cellKsome{}\\
		\hline
		\cite[Fig.~2]{GMR07} & \columnRef{\cite{GMR07} & }1 & 1 & 0 & 2 & 3 & 3 \LongVersionTable{& 41} & 0.05 & \cellKsome{}\\
		\hline
		Web privacy problem \cite{BCLR15} & \columnRef{\cite{BCLR15} & }1 & 2 & 2 & 2 & 4 & 3 \LongVersionTable{& 105} & 0.07 & \cellKsome{}\\
		\hline
		\stylebench{Coffee} & \columnRef{\cite{Andre18FTSCS} & }1 & 2 & 3 & 2 & 5 & 4 \LongVersionTable{& 85} & 0.10 & \cellKall{}\\
		\hline
		\stylebench{Fischer-HSRV02} & \columnRef{\cite{HRSV02} & }3 & 2 & 2 & 6 & 5 & 3 \LongVersionTable{& 2495} & 7.53 & \cellKsome{}\\
		\hline
		\stylebench{STAC:3:v} & \columnRef{\todo{Reference} & }\multicolumn{2}{c |}{\nbLoC{87}} & 2 & 2 & 3 & 9 \LongVersionTable{& 361} & 0.93 & \cellKsome{}\\
		\hline
	\end{tabular}
	
	\label{table:parametric}
\end{table}
\subsection{``Repairing'' a non-opaque PTA}\label{ss:rendering-opaque}

Our method gives a result in time of a union of polyhedra over the internal timing parameters and the execution time.
On the one hand, we believe tuning the internal timing parameters should be easy: for a program, an internal timing parameter can be the duration of a \stylecode{sleep}, for example.
On the other hand, tuning the execution time of a program may be more subtle.
A solution is to enforce a minimal execution time by adding a second thread in parallel with a \stylecode{Wait()} primitive to ensure a minimal execution time.
Ensuring a \emph{maximal} execution time can be achieved with an exception stopping the program after a given time; however there is a priori no guarantee that the result of the computation is correct.
	\ea{Sun Jun, your opinion on this point would be welcome; can we have a ``Wait'' in parallel with the program, to enforce it terminates after at least $n$ time units? Or an exception mechanism to enforce it terminates in exactly $n$ time units?}

\todo{good example: 11A}

\section{Conclusion}\label{section:conclusion}

\LongVersion{In this work, we}\ShortVersion{We} proposed an approach based on parametric timed model checking to not only decide whether the model of a timed system can be subject to timing information leakage, but also to \emph{synthesize} internal timing parameters and execution times that render the system opaque.
We implemented our approach in a framework based on \imitator{},
and performed experiments on case studies from the literature and from a library of Java programs.
\LongVersion{

}%
We now discuss future works\LongVersion{ in the following}.

\paragraph*{Theory}
We proved decidability of the timed-opacity computation problem for TAs, but we only provided an upper bound (EXPSPACE) on the complexity.
It can be easily shown that this problem is at least PSPACE, but the exact complexity remains to be exhibited.

\LongVersion{
In addition, the decidability of the one-clock case remains open:
	that is, is the timed-opacity emptiness problem decidable for PTAs containing a single clock?
Our method consists in duplicating the automaton and adding a clock that is never reset, thus resulting in a PTA with 3 clocks, for which reachability-emptiness is undecidable~\cite{AHV93}.
However, since one of the clocks is never reset, and since the automaton is structurally constrained (it is the result of the composition of two copies of the same automaton), decidability might be envisioned.
}

Finally, while we proved for the class of L/U-PTAs the decidability of the timed-opacity emptiness problem, \ie{} the non-existence of a valuation for which the system is opaque, our result does not necessarily mean that \emph{exact (complete) synthesis} is possible.
In fact, some results for L/U-PTAs were proved to be such that the emptiness is decidable but the synthesis is intractable: that is notably the case of reachability-emptiness, which is decidable~\cite{HRSV02} while synthesis is intractable~\cite{JLR15}. \ea{add minimal-time @ TACAS 2019 + paper testing Lars Luthman}%
Therefore, studying the timed-opacity synthesis problem remains to be done for L/U-PTAs.

\paragraph*{Applications}

The translation of the STAC library required some non-trivial creativity: \LongVersion{while the translation from programs to quantitative extensions of automata is orthogonal to our work, }proposing automated translations of (possibly annotated) programs to timed automata dedicated to timing analysis is on our agenda.

\LongVersion{
	Adding probabilities to our framework will be interesting, helping to quantify the execution times of ``untimed'' instructions in program with a finer grain than an interval; also note that some benchmarks make use of probabilities (notably \stylebench{STAC:18:v}).

	Finally, \imitator{} is a general model checker, not specifically aimed at solving the problem we address here.
	Notably, constraints managed by PPL contain all variables (clocks, timing parameters, and parameters encoding symbolic variables of programs), yielding an exponential complexity.
	Separating certain types of independent variables (typically parameters encoding symbolic variables of programs, and other variables) should highly increase efficiency.
}

\section*{Acknowledgements}
We \LongVersion{would like to }thank Sudipta Chattopadhyay for helpful suggestions, Jiaying Li for his help with preliminary model conversion,
and a\LongVersion{n anonymous} reviewer for suggesting \cref{remark:R2}.
\ifdefined\VersionLong
	\newcommand{\CCIS}{Communications in Computer and Information Science}
	\newcommand{\ENTCS}{Electronic Notes in Theoretical Computer Science}
	\newcommand{\FI}{Fundamenta Informormaticae}
	\newcommand{\FMSD}{Formal Methods in System Design}
	\newcommand{\IJFCS}{International Journal of Foundations of Computer Science}
	\newcommand{\IJSSE}{International Journal of Secure Software Engineering}
	\newcommand{\IPL}{Information Processing Letters}
	\newcommand{\JLAP}{Journal of Logic and Algebraic Programming}
	\newcommand{\JLC}{Journal of Logic and Computation}
	\newcommand{\LMCS}{Logical Methods in Computer Science}
	\newcommand{\LNCS}{Lecture Notes in Computer Science}
	\newcommand{\RESS}{Reliability Engineering \& System Safety}
	\newcommand{\STTT}{International Journal on Software Tools for Technology Transfer}
	\newcommand{\TCS}{Theoretical Computer Science}
	\newcommand{\ToPNoC}{Transactions on Petri Nets and Other Models of Concurrency}
	\newcommand{\TSE}{IEEE Transactions on Software Engineering}

	\renewcommand*{\bibfont}{\small}
	\printbibliography[title={References}]

\else
	\bibliographystyle{splncs04} %
	\newcommand{\CCIS}{CCIS}
	\newcommand{\ENTCS}{ENTCS}
	\newcommand{\FI}{FI}
	\newcommand{\FMSD}{FMSD}
	\newcommand{\IJFCS}{IJFCS}
	\newcommand{\IJSSE}{IJSSE}
	\newcommand{\IPL}{IPL}
	\newcommand{\JLAP}{JLAP}
	\newcommand{\JLC}{JLC}
	\newcommand{\LMCS}{LMCS}
	\newcommand{\LNCS}{LNCS}
	\newcommand{\RESS}{RESS}
	\newcommand{\STTT}{STTT}
	\newcommand{\TCS}{TCS}
	\newcommand{\ToPNoC}{ToPNoC}
	\newcommand{\TSE}{TSE}

	\bibliography{PTA}

\fi

\LongVersion{
\newpage
\appendix

\section{The code of the Java example}\label{appendix:Java}

\begin{lstlisting}[language=Java]
import java.io.BufferedReader;
import java.io.IOException;
import java.io.InputStreamReader;
import java.io.PrintWriter;
import java.net.ServerSocket;
import java.net.Socket;

//Coefficients are Disregarded
public class Category1_vulnerable {
    private static final int port = 8000;
    private static final int secret = 1234;
    private static final int n = 32;
    private static ServerSocket server;

    private static void checkSecret(int guess) throws InterruptedException {
        if (guess <= secret) {
            for (int i = 0; i < n; i++) {
                for (int t = 0; t < n; t++) {
                    Thread.sleep(1);
                }
            }
        } else {
            for (int i = 0; i < n; i++) {
                for (int t = 0; t < n; t++) {
                    Thread.sleep(2);
                }
            }
        }
    }

    private static void startServer() {
        try {
            server = new ServerSocket(port);
            System.out.println("Server Started Port: " + port);
            Socket client;
            PrintWriter out;
            BufferedReader in;
            String userInput;
            int guess;
            while (true) {
                client = server.accept();
                out = new PrintWriter(client.getOutputStream(), true);
                in = new BufferedReader(new InputStreamReader(client.getInputStream()));

                userInput = in.readLine();
                try {
                    guess = Integer.parseInt(userInput);
                    if(guess < 0) {
                        throw new IllegalArgumentException();
                    }
                    checkSecret(guess);
                    out.println("Process Complete");
                } catch (IllegalArgumentException | InterruptedException e) {
                    out.println("Unable to Process Input");
                }
                client.shutdownOutput();
                client.shutdownInput();
                client.close();
            }
        } catch (IOException e) {
            System.exit(-1);
        }
    }

    public static void main(String[] args) throws InterruptedException {
        startServer();
    }
}
\end{lstlisting}

}

\end{document}